\documentclass[12pt]{article}
\usepackage{amscd, amsmath}
\usepackage{amsthm, amssymb}
\usepackage{mathtools}
\usepackage{sectsty}
\usepackage[bottom,flushmargin]{footmisc}
\usepackage{mathrsfs}
\usepackage{float}
\usepackage{array}
\usepackage{accents}
\usepackage{empheq}
\usepackage[dvipsnames]{xcolor}
\usepackage[most]{tcolorbox}
\usepackage[british]{babel}
\usepackage{tikz-cd} 
\usepackage{hyperref}
\usepackage{tocloft}
\usepackage{titlesec}
\usepackage{cite} 
\usepackage{lmodern}
\usepackage[page]{appendix}

\titleformat*{\section}{\fontsize{16}{19}\bfseries\selectfont}
\titleformat*{\subsection}{\fontsize{13}{17}\bfseries\selectfont}
\subsubsectionfont{\normalfont\itshape}

\let\appendixpagenameorig\appendixpagename
\renewcommand{\appendixpagename}{\Large\appendixpagenameorig}

\renewcommand{\baselinestretch}{1.2}
\newcommand{\linebr}{\\[0.2em]}

\DeclareMathAlphabet{\mathpzc}{OT1}{pzc}{m}{it}

\definecolor{myblue1}{RGB}{0, 0, 139}

\hypersetup{
    linktocpage,
    colorlinks,
    citecolor=myblue1,
    filecolor=black,
    linkcolor=myblue1,
    urlcolor=black
}

\newtheorem{theorem}{Theorem}[section]
\newtheorem{lemma}[theorem]{Lemma}
\newtheorem{proposition}[theorem]{Proposition}
\newtheorem*{proposition*}{Proposition}
\newtheorem*{consequence}{Physical implication}

\newtheorem*{assumption(A)}{Condition (A)}

\theoremstyle{remark}

\newtheorem*{remark}{{\bf Remark}}

\numberwithin{equation}{section}

\newcommand{\Lie}{\mathcal L}
\newcommand{\VV}{\mathscr V}
\newcommand{\HH}{\mathcal H}
\newcommand{\NN}{\mathcal N}

\newcommand{\tg}{\tilde{g}}
\newcommand{\kil}{\xi}

\newcommand{\bv}{{\mathbf{v}}}
\newcommand{\bp}{{\mathbf{p}}}
\newcommand{\bu}{{\mathbf{u}}}

\newcommand{\dd}{{\mathrm d}}
\newcommand{\vol}{\mathrm{vol}}
\newcommand{\Vol}{\mathrm{Vol}}

\newcommand*\DAlembert{\mathop{}\!\mathbin\Box}
\newcommand{\Diff}{{\rm Diff}}

\newcommand{\beq}{\begin{equation}}
\newcommand{\eeq}{\end{equation}}
\newcommand{\bal}{\begin{align}}
\newcommand{\eal}{\end{align}}

 \newtcolorbox{empheqboxed}{
 opacityback=0,
 enhanced jigsaw,
 width=\textwidth,
 boxrule=.5pt,
 sharpish corners,
 left=0pt,
 right=2pt,
 top=-9pt, % default value 2mm
 bottom=3pt
}

\newcommand{\LL}{{\text{\scalebox{.9}{$L$}}}}
\newcommand{\RR}{{\text{\scalebox{.9}{$R$}}}}
\newcommand{\PPP}{{\text{\scalebox{.8}{$P$}}}}
\newcommand{\MMM}{{\text{\scalebox{.8}{$M$}}}}
\newcommand{\KKK}{{\text{\scalebox{.8}{$K$}}}}
\newcommand{\Szero}{{\text{\scalebox{.9}{$0$}}}}
\newcommand{\SSzero}{{\text{\scalebox{.7}{$0$}}}}
\newcommand{\MT}{{\text{\scalebox{.9}{${\sf MT^{-1}}$}}}}

\makeatletter

\def\blfootnote{\xdef\@thefnmark{}\@footnotetext}
\makeatother

\begin{document}

\begin{titlepage}

\title{\vspace{-.5cm} \bf Test particles  in Kaluza-Klein models}

\author{{Jo\~ao Baptista}}  
\date{June 2024}

\maketitle

\thispagestyle{empty}
\vskip 10pt
{\centerline{{\large \bf{Abstract}}}}
\noindent
Geodesics in general relativity describe the behaviour of test particles in a gravitational field. In 5D Kaluza-Klein, geodesics reproduce the Lorentz force motion of particles in an electromagnetic field. This paper studies geodesic motion on a higher-dimensional $M_4 \times K$ with background metrics encoding general 4D gauge fields and Higgs-like scalars. It shows that the classical mass and charge of a test particle become variable quantities when the geodesic traverses regions of spacetime with massive gauge fields, such as the weak force field, or with non-constant Higgs scalars. This agrees with the physical fact that interactions mediated by massive bosons can change the mass and charge of particles. The variation rates of mass and charge along a geodesic are given by natural geometric formulae. In regions where mass is preserved, there are additional constants of motion, one for every abelian or simple summand in the Killing algebra of $K$. The last part of the paper discusses traditional difficulties of Kaluza-Klein models, such as the low $q \mkern 1mu / \mkern 1mu m$ ratios in the 5D model. It suggests possible ways to circumvent them. It also remarks the naturalness of a model in which elementary particles always travel at the speed of light in higher dimensions.

\vspace{-1.0cm}
\let\thefootnote\relax\footnote{
\noindent
{\small {\sl \bf  Keywords:} Kaluza-Klein theories; geodesic motion; Riemannian submersions; test particles; momentum; rest mass; charge.}
}

\end{titlepage}

\pagenumbering{roman}

\renewcommand{\baselinestretch}{1.14}\normalsize
\tableofcontents
\renewcommand{\baselinestretch}{1.18}\normalsize

\newpage 

\pagenumbering{arabic}

\section{Introduction and overview of results}

\noindent
This paper studies geodesic motion on a higher-dimensional spacetime $M_4 \times K$ and how it is perceived through the lens of 4D physics. The setting is more general than usual, with background Kaluza-Klein metrics that encode both massless and massive 4D fields. Geometrically, those metrics determine Riemannian submersions with fibres that need not be totally geodesic. The main physical motivation is the notion that, in these conditions, a particle’s classical mass and electromagnetic charge should not be strict constants of motion. When massive gauge fields are present, or when the Higgs fields are not constant, physical particles can interact with massive bosons, so their own charge and mass can change. Thus, it would be natural to have definitions of mass and charge of a test particle that enabled such variations along a geodesic. This would also illustrate how a higher-dimensional, classical theory can produce 4D effects (particle charge and mass change) that usually only appear in QFT calculations.

Thinking along those lines, the author was led to consider how 4D proper time, denoted $\tau$, evolves along an affine geodesic $\gamma(s)$ on the spacetime $P = M_4\times K$. In this setting, the covariantly conserved momentum is the higher-dimensional tangent vector
\beq \label{MassDefinition0}
p(s)  \ = \  \sigma \, \frac{\dd \gamma}{\dd s}  \ =  \   \sigma \, \frac{\dd \tau}{\dd s}\; \frac{\dd \gamma}{\dd \tau} \ = \ m(s) \; \frac{\dd \gamma}{\dd \tau}  \ , 
\eeq
where $\sigma$ is a constant with dimension $\MT$. The factor $m(s) = \sigma \, \frac{\dd \tau}{\dd s}$ is the particle's rest mass, as can be recognized by projecting the relation to $M_4$. So a non-constant rate of change $\frac{\dd \tau}{\dd s}$ is perceived in 4D as a variable mass. As will be seen later, it turns out that 4D proper time and the geodesic parameter $s$ need not be proportional to each other in regions where the internal geometry is not covariantly constant. So the particle's classical mass $m(s)$ can vary when the geodesic traverses those regions, as desired.

It also turns out that the phenomenon of geodesic mass variation had been previously noticed in the Kaluza-Klein literature. Specifically, in studies of the 5D model, with $K= \mathrm{U}(1)$ and mass variation coming from a non-constant size of the internal circle \cite{Kovacs, GK, CP, CE, LM}. In that electromagnetic setting, however, it was regarded as an undesirable difficulty of Kaluza-Klein theory. Rest mass variation is not observed in real electromagnetic interactions, so it was not easy to offer a physical understanding of the phenomenon. Arguments were proposed to avoid it or find alternative interpretations \cite{SW, Mi}. Another example was described in \cite{Feher}, where $K$ is a group manifold with variable size and fixed, bi-invariant geometry. Still, the physical meaning of mass variation remained unclear.

A wider perspective seems to be key in this instance. For a higher-dimensional $K$ with richer geometry, it becomes possible to distinguish the separate effects on geodesics caused by electromagnetic-like gauge fields, non-abelian massless gauge fields, massive gauge fields and Higgs-like scalars. And when one recognizes that the only fields able to change the mass of a test particle are the massive gauge fields (such as the weak force field) and the Higgs-like scalars, a natural physical interpretation emerges. The phenomenon of geodesic mass variation seems to describe the physical fact that interactions mediated by massive fields and massive bosons are able to change the mass of particles.

What about charge? In the geodesic model, all internal motion is perceived as mass in four dimensions, but electromagnetic charge is related to the component of internal momentum along a specific Killing direction inside $K$. So, in this model, there can be mass without charge, but not the other way around, which fits observation. The independent existence of mass and charge can occur only when the dimension of $K$ is greater than one. Otherwise, there is no meaningful distinction between the internal momentum and its component along the Killing field.

The higher-dimensional definition of charge of a test particle has three natural properties: {\it i)} charge is a constant of geodesic motion in regions of spacetime that only have massless gauge fields and constant Higgs scalars; {\it ii)} in those regions, geodesics project down to Lorentz force motions in four dimensions, as usual in 5D Kaluza-Klein; {\it iii)} in regions of spacetime with massive gauge fields, the test particle's charge can vary along a geodesic if and only if the bosons associated with the fields are charged.
Thus, higher-dimensional geodesics give a fairly good account of the properties of charged particles, at least qualitatively. They reflect the physical fact that interactions mediated by charged gauge bosons, such as the Standard Model’s W bosons, can modify the charge of particles, while interactions mediated by massless or by massive, neutral gauge bosons, cannot. This property of higher-dimensional geodesics apparently has not been remarked before. 

Another relevant property of the geodesic model is that particles travelling at the speed of light on $M_4$ cannot interact directly with gauge fields or Higgs-like scalars, at least classically. More precisely, the projection to $M_4$ of the particle's motion on $P$ is independent of the values of such fields.  This property is incompatible with the motion of the old massless neutrinos, since they were thought to interact with the weak field and travel at the speed of light on Minkowski space. So the geodesic model disfavours the existence of massless neutrinos or similar particles. This also does not seem to have been remarked before. Of course, one should still bear in mind that geodesic motion is a classical approximation that disregards back-reaction, quantization and the fermionic nature of particles.

Our study of geodesic motion on $P = M_4 \times K$, for arbitrary $K$,  relies on classical geometrical results about Riemannian submersions developed in \cite{Ehresmann, Hermann, ONeill1, ONeill} and presented in \cite{Besse, FIP}, for example. Extended reviews of the Kaluza-Klein framework can be found in \cite{Bailin, Bleecker, Bou, CFD, CJBook, Duff, WessonOverduin, Witten81}. Some of the early original references are \cite{Kaluza, Klein, EB, Jordan, Thiry, DeWitt, Kerner, Traut, Cho}. Geodesic motion has been studied since the very beginning of Kaluza-Klein literature. The main focus has been on the 5D case, where calculations are more explicit and differential geometric techniques less necessary (e.g. \cite{Kaluza, Kerner, LR, GK, DO, FH, WessonPonce, LMP}). This paper follows the notation in \cite{Bap} and its treatment of massive gauge fields and Higgs-like fields. 
%A shorter account of the results on rest mass variation is given in the note \cite{Bap3}. 
More comments about the literature will be added below, as we give an overview of the main results and observations in this paper.

\subsubsection*{\bf Rest mass variation}

Consider a metric $g_P$ on the higher-dimensional spacetime $P=M_4\times K$ such that the projection to $M_4$ is a Riemannian submersion. As described in section \ref{NullCurves}, this is equivalent to taking a $g_P$ determined by three simpler objects: {\it i)} a metric $g_M$ on the base $M_4$; {\it ii)} a family of Riemannian metrics $g_K(x)$ on the fibres $K$ parametrized by the points  $x \in M_4$; {\it iii)} a gauge one-form $A$ on $M_4$ with values in the vector fields on $K$. These objects determine the higher-dimensional metric through the relations
\begin{align*}
g_P (U, V) \ &= \ g_K (U, V) \nonumber \\
g_P (X, V) \ &= \  - \ g_K \left(A (X), V \right) \nonumber \\
g_P (X, Y) \ &= \ g_M (X, Y) \ + \  g_K \left(A(X) , A(Y) \right) \ ,
\end{align*}
valid for all tangent vectors $X,Y \in TM$ and vertical vectors $U, V \in TK$. These relations generalize the usual Kaluza ansatz for $g_P$. Since the gauge one-form has values in the Lie algebra of vector fields on $K$, the gauge group is $\Diff(K)$ or a subgroup. The gauge group need not act on $K$ only through isometries of $g_K$.

Now let $\gamma (s)$ be a timelike or null geodesic on $(P, g_P)$ representing the motion of a test particle. It is a curve satisfying $\nabla_{\dot{\gamma}} \dot{\gamma} = 0$. Denote by $\gamma_M (s)$ the projection of this curve to Minkowski space.  
The main result of section \ref{MassVariation} says that the particle's rest mass, as  defined in \eqref{MassDefinition0}, changes according to 
\beq \label{DerivativeMass}
c^2\; \frac{\dd}{\dd s} \; m^2(s)  \ = \ - \,  (\dd^A g_K)_{\dot{\gamma}_M} (p^\VV, p^\VV) \ . 
\eeq
Here $p^\VV$ denotes the vertical (internal) component of the particle's momentum vector, while $(\dd^A g_K)_{\dot{\gamma}_M}$ is a covariant derivative of the internal metric along the vector $\dot{\gamma}_M$ tangent to $M_4$. So the rest mass of the particle is a constant of geodesic motion in regions where $g_K$ is covariantly constant, but may change elsewhere. 

The derivative $\dd^A g_K$ measures how $g_K$ changes along $M_4$ up to diffeomorphisms of $K$. It is equivariant under $\Diff (K)$-gauge transformations. Geometrically, it can be identified with the second fundamental form of the fibres of $P$. As in \cite{Bap}, it can be expressed in terms of Lie derivatives and the gauge one-forms $A^a$ as 
\beq \label{DefinitionCovariantDerivative}
(\dd^A g_K)_X (U, V) \ = \  (\Lie_{X}\, g_K) (U, V) \ + \ A^a (X) \; (\Lie_{e_a}\, g_K) (U, V) \ .
\eeq
Here $X$ is any vector in $TM$; $U$ and $V$ are vertical vectors in $TK$; and $\Lie_{e_a}\, g_K$  denotes the Lie derivative along the internal vector field $e_a$. Thus, $\dd^A g_K$ vanishes when the internal metric is constant along $M_4$ and, simultaneously, only gauge fields with values in the Killing vectors of $g_K$ are non-zero. Those are precisely the massless gauge fields, since 
\beq   \label{MassFormula}
\left(\text{Mass} \ A_\mu^a \right)^2 \ \ \propto \ \ \frac{ \int_K  \; \left\langle \Lie_{e_a}\, g_K,  \; \Lie_{e_a}\, g_K \right\rangle  \, \vol_{g_\KKK} }{ 2 \int_K  \,  g_K (e_a ,  e_a  ) \ \vol_{g_\KKK} } 
\eeq
for any (divergence-free) vector field $e_a$ on a compact $K$ \cite{Bap}. Therefore, we conclude that when the geodesic traverses regions of spacetime with changing $g_K$ or with massive gauge fields, such as the weak force field, the particle's classical rest mass $m(s)$ can indeed vary. 

Rest mass variation can be understood as a transfer between 4D momentum and internal momentum. This transfer does not occur in regions with vacuum-type, product geometries $g_M + g_K$. There, the two momentum components are uncoupled and separately conserved. If massless gauge fields are turned on, the geometry is no longer a product, but 4D motion is only lightly coupled to internal motion. The direction of 4D momentum can change, but its magnitude is still conserved. These are Lorentz force-type motions. Internal and 4D momenta will just rotate within the vertical and horizontal subspaces, respectively. In regions where the higher-dimensional geometry is very different from a product, only the full momentum vector conserves its norm along general geodesics. The norms of the 4D and internal components can both change, with opposite signs, due to a rotation of the full momentum mixing those components. This process is perceived in 4D as a change of the particle's rest mass. The geometry distortions that produce it correspond to a $g_P$ encoding 4D massive gauge fields or a non-constant internal metric.

Note that having an internal metric that changes along $M_4$ is equivalent to having non-constant Higgs-like fields. If we regard the components of $g_K$ as fields on $M_4$, they play the exact role of Higgs fields in usual gauge theories. This is apparent from the decomposition of the higher-dimensional scalar curvature $R_{g_\PPP}$ in the Einstein-Hilbert action, 
 \beq \label{GaugedSigmaModelAction0}
\int_P \, R_{g_\PPP}  \vol_{g_\PPP}   \ = \ \int_P \, \Big[\, R_{g_\MMM}  \, + \, R_{g_\KKK} \, - \, \frac{1}{4}\, |F_A|^2 \, - \,  \frac{1}{4}\,  |\dd^A  g_K|^2  \, + \,  |\dd^A \, (\vol_{g_\KKK})|^2 \, \Big] \,\vol_{g_\PPP} \, . 
\eeq
 This formula extends the usual Kaluza-Klein result to the setting of general Riemannian submersions, where the geometry of the fibres can change. For more details, see \cite{Bap}.

\subsubsection*{\bf Charges and charge variation}

Now let us consider particles' charges. Let $\kil$ be an electromagnetic-like Killing vector field of $g_K$. It is an internal Killing field that commutes with all other Killing fields on $K$. If a particle's motion is parameterized by a geodesic $\gamma (s)$ on $P$ with momentum $p = \sigma\, \dot{\gamma}$, we define the particle's charge with respect to $\kil$ as the scalar 
\beq
q_\kil (s) \ := \ -\, g_P(\kil,\, p) \ = \ -\, g_K(\kil, \, p^\VV) \ .
\eeq
Section \ref{MasslessGaugeFields} shows that $q_\kil (s)$ is a constant of geodesic motion in regions where only massless gauge fields are present and the internal geometry does not change. This is true for any metric $g_M$ on $M_4$. But there are more constants of motion in these regions. Essentially, there is one constant for every summand in the decomposition
\beq
\mathfrak{k} \ = \ \mathfrak{a}_1 \oplus \cdots \oplus \mathfrak{a}_m
\eeq
of the Killing algebra of $g_K$ as a sum of abelian or simple subalgebras. That constant of motion measures how orthogonal the derivative vector $\dot{\gamma} (s)$ is to the subspace of $T_{\gamma (s)} P$ spanned by the Killing fields in the respective summand. And, for a geodesic, this scalar is independent of $s$. %The Killing algebra of $g_K$, recall, determines the massless gauge fields in the model.

In regions where the internal metric $g_K$ is constant but some massive gauge fields are non-zero, the charges described above are well-defined but are not constants of motion anymore. For example, section \ref{ChargeVariation} shows that an electromagnetic-like charge evolves as 
\beq
\frac{\dd}{\dd s} \; q_\kil (s) \ = \  A^a(\dot{\gamma}_M) \  g_P ( [ \kil , \, e_a], \, p )  
\eeq
along a higher-dimensional geodesic $\gamma (s)$ with momentum $p(s)$. So the test particle's $\kil$-charge may vary when $A^a$ is non-zero and the associated gauge boson is $\kil$-charged (i.e. when $[ \kil , \, e_a] \neq 0$ as a vector field on $K$). This effect apparently has not been reported before. It agrees with the physical fact that interactions mediated by massive, charged gauge bosons, such as the Standard Model's W bosons, can modify the charge of particles.

The charge variation formula can be extended to the scalars associated to the simple summands $\mathfrak{a}_r$ in the Killing algebra of $g_K$. This variation is different from the rotation of isospin along a geodesic found in \cite{Kerner, Wong} and reviewed in \cite{HZ}. More details in section \ref{ChargeVariation}.

\subsubsection*{\bf Equations of geodesic motion}

In section \ref{GeneralGeodesicsSubmersions} we write down the equation of geodesic motion for a general submersive metric on $M_4 \times K$, as translated from \cite{ONeill}. The differences in notation between \cite{ONeill} and the present paper are described in appendix \ref{Auxiliary results}. The horizontal component of that equation says that the projection of the geodesic to four dimensions, denoted $\gamma_M(s)$, satisfies
\beq  \label{GeneralizedLorentzForce0}
g_M \big(\nabla^M_{\dot{\gamma}_M}\, \dot{\gamma}_M , \, X \big) \ = \ g_K (e_a, \, \dot{\gamma}^\VV) \, F_{A}^a (\dot{\gamma}_M , X ) \ + \ \frac{1}{2} \; (\dd^A g_K)
_X (\dot{\gamma}^\VV, \dot{\gamma}^\VV) \ .
\eeq
This is a generalization of the equation derived by Kerner when $K$ is a Lie group with a constant bi-invariant metric \cite{Kerner}. It also extends the equations obtained when $K$ is a circle or group manifold of variable size \cite{Kovacs, GK, CP, LM, Feher}. It is valid for arbitrary gauge fields, arbitrary $K$, and internal metrics $g_K(x)$ that can vary arbitrarily along $M_4$. 

In regions where there are no massive gauge fields and the internal geometry is constant, the term $\dd^A g_K$ vanishes and the equation simplifies. It reduces to the usual Lorentz force equation when only an electromagnetic-like gauge field is present, as in the 5D calculation. This is verified in section \ref{ChargeVariation} using the previous definitions of mass and charge. It justifies the interpretation of the scalar $q_\kil (s)$ as a charge.

\subsubsection*{\bf A unique speed in higher dimensions}

In section \ref{AUniqueSpeed}, we remark the naturalness of the hypothesis that elementary particles always travel at the speed of light in higher dimensions. It is the projection of velocities to three dimensions that appears to produce speeds in the range $[ 0, c]$, as observed macroscopically. This is equivalent to saying that particles always follow null paths on $P$.

This hypothesis is not  entirely unreasonable because null paths on $P$ always project down to timelike or null paths on $M_4$, in the Kaluza-Klein framework. They never project down to spacelike paths. Higher-dimensional null paths can cover all types of causal motion on Minkowski space. So timelike paths on $P$ do not seem necessary. Thus, it is natural to forgo them provisionally and investigate the consistency of a dynamical model entirely based on null paths on $P$.

In these conditions, all particles obey an energy-momentum relation similar to that of photons, but in higher dimensions. It projects down to the usual energy-momentum relation in four dimensions. For example, when $g_P = g_M + g_K$ is a simple product metric, the higher-dimensional momentum vector can be written in an inertial frame as $p = (E, \, \bp + p_K)$, where $E$ is the particle's energy, $\bp$ and $p_K$ are its 3-momentum and internal momentum vectors, respectively, and we have put $c=1$. So when $p$ is null with respect to $g_P$, we get that
\[
E^2 \ = \ g_P(\bp + p_K, \,\bp + p_K) \ = \ |\bp|^2 \, + \, g_K(p_K,\,p_K) \ .
\]
 The first equality is a photon-like energy-momentum relation in higher dimensions. The second equality becomes the usual 4D energy-momentum relation if the particle's rest mass is identified with the norm of its internal momentum, $m^2 = g_K(p_K,\,p_K) $.

If the higher-dimensional speed is always $c$, as advocated, a particle at rest on 3D space is necessarily moving at full speed $c$ along $K$. Then the associated kinetic energy is a natural source of the particle's energy at spatial rest, $E_0 =m c^2$. It is appealing to think that rest energy is simply the kinetic energy of internal motion, with no need for alternative mechanisms to store energy in a point-like mass. 

The hypothesis of a unique speed in higher dimensions provides a natural origin for 3D rest energy. Conversely, the assumption of a fully kinetic origin of rest energy, if granted, also implies that timelike geodesics in higher-dimensions should not be allowed to represent physical motions. Otherwise, a timelike particle moving with a 3-dimensional speed lower than $c$ could very well have zero velocity along $K$ and hence have no rest energy or rest mass. And such particles have never been observed. Thus, when 3D rest energy is identified with the internal kinetic energy, only allowing null geodesics on $P$ derives from the experimental fact that massless particles always travel at the speed of light on $M_4$. Having a unique finite speed $c$ for all elementary particles is also a mathematically attractive feature, simpler than having a closed interval $[0, c]$ of possible speeds. The inescapable price is having to work with a higher-dimensional spacetime, of course.

As described in section \ref{AUniqueSpeed}, the hypothesis of a unique speed in higher dimensions for all elementary particles would not be tenable in the traditional 5D Kaluza-Klein model. Even forgetting about the strong and weak forces. This is because the null geodesics of a 5D metric with circle isometry do not have enough degrees of freedom to reproduce all the combinations of mass and electromagnetic charge observed in elementary particles. This objection disappears for a higher-dimensional $K$. So a geodesic model entirely based on null paths seems to be more tenable in this case. It also fits better with the commonly stated aim of representing fermions by solutions of a single, massless, Dirac-like equation for higher-dimensional spinors.

\subsubsection*{\bf Spaces of null geodesics}

Let $\NN_h^+$ denote the space of null geodesics starting at a point $h$ in $P$ and moving forward in time. Each geodesic $\gamma (s)$ in this space is characterized by its null tangent vector $\dot{\gamma} (0)$ at $h$. In section \ref{SectionSpaceGeodesics} we describe two distinct parameterizations of $\NN_h^+$, one relying on the particles' momenta and the other on the celestial sphere of velocities.

Different geodesics in $\NN_h^+$ represent the motion of particles with different masses and electromagnetic charges. Fixing the values of those constants carves out a smaller subspace $\NN_h^+(m, q_\kil)$ inside $\NN_h^+$. In section \ref{SectionSpaceGeodesics} it is shown that
\beq \label{SpaceGeodesics1}
\NN_h^+(m, q_\kil) \ \simeq \ \begin{cases}
\emptyset               & \text{if }\   |q_\kil|  \; > \;   m\, c\, |\kil|   \\
\mathbb{R}^{3}     & \text{if }\  |q_\kil| \; =\;   m\, c\, |\kil|  \\   
 \mathbb{R}^{3} \times S^{k-2} & \text{if } \ |q_\kil| \ <\  m\, c\, |\kil|   \ .
\end{cases}
\eeq
Here $|\kil|$ denotes the Riemannian length $\sqrt{(g_K)_h (\kil,\kil)}$ of the Killing vector field at the point $h$. Thus, in this model, particles with a given classical mass cannot have arbitrarily strong charge. This is natural because, according to \eqref{MassParticle}, mass is related to the norm of vertical momentum, while $q_\kil$ measures the component of that same momentum along $\kil$.

\subsubsection*{\bf Some traditional difficulties in Kaluza-Klein}

Kaluza-Klein models are often studied under strong simplifying assumptions, such as minimal 5D dimensions, or constant internal geometry, or the assumption that all relevant gauge fields are associated with internal isometries. Those simplifications facilitate the analysis but also create problems. In fact, we argue that some of the difficulties traditionally attributed to the Kaluza-Klein framework are due to those simplifications.

For example, the physical weak force field is usually associated with an $\mathrm{SU}(2)$-isometry of internal space. But its bosons are massive, albeit light when compared to the Planck mass. So the mass formula \eqref{MassFormula} suggests that the weak field should not be associated with exact isometries of $g_K$. The Lie derivatives $\Lie_{e_a}g_K$ can be small yet non-zero. Abandoning the exact isometry assumption has an extra advantage, as described in \cite{Bap}. It offers a possible way out of the main no-go arguments against having chiral fermions in Kaluza-Klein, such as the arguments based on the Atiyah-Hirzebruch theorem \cite{Witten83}.

Another often-cited difficulty, in the 5D model, is that the Lorentz force equation of motion can be recovered from geodesics only in regions where the internal circle has constant size. And if this condition is granted, the full 5D equations of motion force the norm $|F_A|^2$ of the electromagnetic field strength to vanish in those same regions, which is not realistic \cite{WessonOverduin, Duff}. Moreover, even ignoring that problem, the range of 4D motions projected by 5D geodesics has severe limitations. Due to the normalization condition of the Killing field $\kil$, all 5D geodesics project down to Lorentz force motions on $M_4$ with $q/m$ ratios that are too low when compared to the physical ratios for elementary particles \cite{GK, CE}. So timelike or null 5D geodesics cannot describe the 4D motion of charged elementary particles. One needs spacelike 5D geodesics \cite{DO}.

Our first point is that these difficulties are less acute in higher-dimensional models. For example,  for higher-dimensional $K$, after transforming the Lagrangian in \eqref{HDEinsteinHilbertAction} to the Einstein frame, the local constancy of internal volume only implies that\footnote{This equation can be derived from the last equation in section 3.4 of \cite{Bap} after expressing the normalized metric $\bar{g}_K$ of that section in terms of the plain, un-normalized metric $g_K$.}
\beq
\Big( \frac{\kappa_P}{\kappa_M \, \Vol_{g_\KKK}} \Big)^2  g_K(e_a, e_b) \, (F_A^a)^{\mu \nu} \, (F_A^b)_{\mu \nu} \; + \; R_{g_\KKK} \; - \; \frac{2\, \Lambda\, k}{k+2} \ = \ 0 \ ,
\eeq
where $k$ denotes the dimension of $K$. So there is room for $|F_A|^2$ to vary, as long as those changes are compensated by variations of the internal metric $g_K$ that change the scalar curvature $R_{g_\KKK}$ without affecting the total volume. For example TT-deformations of $g_K$. Moreover, these constraints are derived solely from the Einstein-Hilbert action on $P$, which should not tell the whole story for realistic models operating at different scales \cite{Bap}. 

In section \ref{DiscussionDifficulties}, we discuss the second difficulty, the small $q/m$ ratios implied in 5D geodesics. It is shown that, for higher-dimensional $K$, the normalization of the Killing field no longer determines the value of $g_K(\kil, \kil)$. Only the average value of that norm over $K$. This makes the problem less acute, because that average can be significantly different from the point value $g_K(\kil, \kil)$ that appears in the geodesic equation. 
Moreover, the usual normalization condition of $\kil$ assumes that the background metric on $P$ satisfies the higher-dimensional Einstein equations. For different backgrounds, the normalization condition of $\kil$ can be less problematic.

Another point discussed in section \ref{DiscussionDifficulties} is that, if the background metric satisfies the Einstein equations on $P$, then its projection to $M_4$ should not be identified with the physical 4D metric, in general. Only with a rescaled version of it. This is related to the well-known need to transform the dimensionally-reduced Lagrangian from the Jordan frame to the Einstein frame. So one should be careful when studying geodesics on such backgrounds. Different strategies to deal with this issue are discussed in that final section.

\subsubsection*{\bf Conceptual simplicity}

Kaluza-Klein models strive for conceptual unification at the classical level, before thinking about quantization. Traditionally, they mainly deal with the unification of gauge fields and the 4D metric as components of a unique, higher-dimensional metric. Both in abelian and non-abelian gauge theories. It should also be possible to describe spontaneous symmetry breaking as a dynamical process of the internal metric, in which the isometry group of $g_K$ is broken to generate the gauge bosons' mass according to formula  \eqref{MassFormula} \cite{Bap}.

The main message of the present paper, in turn, is that conceptual unification can be taken farther in simple, higher-dimensional models, hopefully without contradicting observation. At the level of test particles and geodesic motion, one can construct a model where massive and massless particles both travel at the speed of light in higher dimensions, satisfying a photon-like energy-momentum relation that projects down to the usual 4D relation on $M_4$. A model where mass, charges, and 4D momentum describe different aspects of a unique higher-dimensional momentum vector, which is covariantly conserved along geodesics. A model where the energy stored in the 3D rest mass of classical particles is simply the kinetic energy of internal motion. 

In this picture, the classical rest mass is not a constant attribute of a test particle. It is a dynamical quantity measuring the internal component of the particle's momentum. It can vary along geodesics if the background geometry is sufficiently distorted away from the vacuum configuration, since this enables transfers between the horizontal and vertical components of momentum. In particular, the geodesic model illustrates how a higher-dimensional, classical theory can exhibit qualitative features (particle charge and mass change) that are usually reserved for QFT calculations. 

The geodesic model also has clear limitations, of course. There is no quantization of charge, mass or energy. Particles do not back-react on fields and their fermionic nature is ignored. As in the case of GR, geodesics offer a very simplified picture of how a particle interacts with fields. But that classical preview could be useful, nonetheless. A natural heuristic picture may help to guide deeper studies of higher-dimensional models.

\section{Motion in the vacuum and 4D inertial frames}
\label{VacuumMotion}

To establish notation, we start by considering motion on $P = M_4 \times K$ equipped with a product metric $g_M + g_K$. Consider an inertial frame on $M_4$ with coordinates $( t, x^1, x^2, x^3)$ and write the Minkowski metric as
\beq \label{MinkowskiMetric}
g_M \ = \ \dd x^1 \otimes \dd x^1  \, +\,  \dd x^2 \otimes \dd x^2 \, + \,  \dd x^3 \otimes \dd x^3 \, - \, c^2\,  \dd t \otimes \dd t \ .   
\eeq
Take arbitrary local coordinates $y^j$ on the internal space $K$. In an inertial frame related by a boost with speed $u$ along the $x^1$-axis, the new coordinates satisfy the usual relations
\begin{align}
\dd x'^1 \ &= \ \frac{\dd x^1  - u \, \dd t}{\sqrt{1 - u^2 /c^2}}   &  \dd t' \ &= \ \frac{\dd t  - u \,\dd x^1 / c^2 }{\sqrt{1 - u^2 /c^2}}   \nonumber \linebr
\dd x'^n \ &= \ \dd x^n   \quad {\rm{for}}\ \  n=2,3    &  \dd y'^j \ &= \ \dd y^j    \ . 
\end{align}
All coordinates transversal to the boost remain invariant, including the internal ones. A particle moving on $M_4 \times K$ can be parameterized by a curve $\gamma(s) = (\gamma_M(s) ,\, \gamma_K(s) )$. Taking the derivative with respect to $s$, we have the tangent vectors
\[
\frac{\dd \gamma}{\dd s}(s) \ = \ \dot{\gamma} \ = \  \dot{\gamma}_M \ + \ \dot{\gamma}_K \ ,
\]
with $\dot{\gamma}_M$ tangent to $M_4$ and $\dot{\gamma}_K$ tangent to $K$. They satisfy the relation
\begin{align}
g_K(\dot{\gamma}_K, \dot{\gamma}_K) \; -\; g_P( \dot{\gamma}, \dot{\gamma}) \; &= \;  -\, g_M(\dot{\gamma}_M, \dot{\gamma}_M)  \ = \ \Big( c \, \frac{\dd \tau}{ \dd s} \Big)^2 \ ,
\end{align}
where $\tau$ denotes the particle's 4D proper time. Since $g_K$ is Riemannian by assumption, the term $g_K(\dot{\gamma}_K, \dot{\gamma}_K)$ is always non-negative. Thus, if $\gamma$ is a timelike curve on $P$, the projection $\gamma_M$ will also be timelike on $M_4$.  If $\gamma$ is null on $P$, the projection $\gamma_M$ will generically be timelike on $M_4$, but can also be null when $\dot{\gamma}_K$ is zero. In principle, even some spacelike curves on $P$ can project down to timelike curves on $M_4$, if the norm $g_K (\dot{\gamma}_K, \, \dot{\gamma}_K)$ is sufficiently large. We will not consider that possibility here. In a Kaluza-Klein model the additional dimensions are interpreted to be physical as well, so they are subject to the same causal restrictions as the Minkowski dimensions. So we will only consider curves such that $g_P( \dot{\gamma}, \dot{\gamma})$ is negative or zero.

A particle's velocity in the inertial frame $( t, x^1, x^2, x^3, y^j)$ is the derivative of its position with respect to the time coordinate. So the internal velocity of the test particle represented by $\gamma (s)$ is given by 
\[
v_K \ = \ \frac{\dd \gamma_K}{\dd t} \ = \  \frac{\dd s}{\dd t} \; \frac{\dd \gamma_K}{\dd s} \ = \  \frac{1}{\dot{\gamma}^0} \; \dot{\gamma}_K \ , 
\]
which is a vector tangent to $K$ at the point $\gamma_K(s)$. Similarly, the particle's 3D velocity in the frame is
\[
\bv \ = \ \frac{\dd \gamma^n}{\dd t} \frac{\partial}{\partial x^n} \ = \  \frac{1}{\dot{\gamma}^0} \; \dot{\gamma}^n \; \frac{\partial}{\partial x^n} \ ,
\]
with an implicit sum over $n=1,2,3$. The dot denotes derivation with respect to $s$.

An affine geodesic on $P$ is a curve $\gamma (s)$ satisfying $\nabla_{\dot{\gamma}} \dot{\gamma} = 0$, where $\nabla$ denotes the Levi-Civita connection of $g_P$. When $g_P$ is a product metric,  the projected curves $\gamma_M (s)$ and $\gamma_K (s)$ are also geodesics on the respective spaces. For more general $g_P$ they are not. Although the geodesic equation is defined here using the Levi-Civita connection of $g_P$, any connection with totally anti-symmetric torsion would lead to the same equation, hence to the same geodesics on $P$.

\section{Curves in Riemannian submersions}
\label{NullCurves}

Take a Lorentzian metric $g_P$ on the higher-dimensional space $P =  M_4 \times K$ such that the projection $\pi: P \rightarrow M_4$ is a Riemannian submersion. As in the usual Kaluza-Klein framework, this metric determines three more familiar objects: 
\begin{itemize}
\item[{\bf i)}]  through projection, a unique Lorentzian metric $g_M$ on $M_4$; 
\item[{\bf ii)}] through restriction to the fibres $\{ x\} \times K$, a family of Riemannian metrics $g_K(x)$ on the internal spaces parameterized by the points in $M_4$;
\item[{\bf iii)}] gauge fields on spacetime, encapsulated in a one-form $A$ on $M_4$ with values in the Lie algebra of vector fields on $K$.
\end{itemize}
The equations linking these objects to the higher-dimensional metric $g_P$ are
\bal \label{MetricDecomposition}
g_P (U, V) \ &= \ g_K (U, V) \nonumber \\
g_P (X, V) \ &= \  - \ g_K \left(A (X), V \right) \nonumber \\
g_P (X, Y) \ &= \ g_M (X, Y) \ + \  g_K \left(A(X) , A(Y) \right) \ ,
\end{align}
for all tangent vectors $X,Y \in TM$ and vertical vectors $U, V \in TK$. These relations generalize the usual Kaluza ansatz for $g_P$. They show how to reconstruct the higher-dimensional metric from the data $(g_M , A, g_K)$. The correspondence between submersive metrics on $P$ and that data is a bijection. This is described in more detail in \cite{Bap}.

Choosing a set $\{ e_a \}$ of independent vector fields on $K$, the one-form on spacetime can be decomposed as a sum 
\beq \label{GaugeFieldExpansion}
A(X) \ = \ \sum\nolimits_a \,A^a(X) \, e_a \ ,
\eeq
where the real-valued coefficients $A^a(X)$ are the traditional gauge fields on $M_4$. For general submersive metrics on $P$ this can be an infinite sum, with $\{ e_a \}$ being a basis for the full space of vector fields on $K$, which coincides with the Lie algebra of the diffeomorphism group ${\rm Diff} (K)$. The curvature $F_A$ is a two-form on $M_4$ with values in the Lie algebra of vector fields on $K$. It can be defined by 
\[
F_{A} (X, Y)  \ := \ (\dd_M A^a) (X, Y) \, e_a  \ + \ A^a (X)\, A^b (Y) \, [e_a, e_b]  \ ,
\]
where the last term is just the Lie bracket $ [A(X), A(Y)] $ of vector fields on $K$.

The tangent space to $P=M_4\times K$ has a natural decomposition $TP = TM \oplus TK$. Since $TK$ is the kernel of the projection $TP \rightarrow TM$, it is also called the vertical sub-bundle $\VV$ of $TP$. The higher-dimensional metric $g_P$ determines an orthogonal complement $\HH \simeq (TK)^\perp$, called the horizontal sub-bundle. So from $g_P$ we get a second decomposition
\beq \label{HorizontalDistribution}
T P \; =\; \HH \oplus \VV   \ .
\eeq
Every tangent vector $E\in TP$ can be written as a sum of components $E^\HH  +  E^\VV$. The relation between the two decompositions of $TP$ is quite simple in a Riemannian submersion.  Writing $E = E_M + E_K$ for the components according to $TP = TM \oplus TK$, we have
\beq \label{DefinitionHorizontalDistribution}
E^\VV  \ = \   E_K \; - \;  A (E_M)    \qquad \qquad  E^\HH  \ = \  E_M \; + \;   A (E_M)   \ .
\eeq
So the information contained in the gauge one-form $A$ on $M$ is equivalent to the information contained in the horizontal distribution $\HH \subset TP$. Geometrically, it is well known that the curvature $F_A$ is the obstruction to the integrability of the distribution $\HH$, in the sense that it vanishes if and only if $P$ can be foliated by horizontal submanifolds whose tangent space coincides with $\HH$ \cite{Besse}.

Now let $\gamma (s)$ be a curve on $P$ parameterized by $s$. Let $\gamma_M (s)$ and $\gamma_K (s)$ denote its projections onto the factors $M_4$ and $K$. The tangent vectors to $P$ obtained by derivation with respect to $s$ have two decompositions 
\beq \label{TwoDecompositions}
\frac{\dd \gamma}{\dd s} \ = \ \dot{\gamma} \ = \  \dot{\gamma}_M \ +\  \dot{\gamma}_K  \ = \  \dot{\gamma}^\VV \; +\ \dot{\gamma}^\HH \ . 
\eeq
According to \eqref{DefinitionHorizontalDistribution} these are related by
\bal \label{DecompositionsTangentCurve}
\dot{\gamma}^\VV \ &= \   \dot{\gamma}_K   \,  - \,  A^a (\dot{\gamma}_M) \, e_a      \linebr
\dot{\gamma}^\HH \ &= \   \dot{\gamma}_M  \, +\,  A^a (\dot{\gamma}_M) \, e_a     \nonumber \ .
\end{align} 
Since $g_P$ restricted to horizontal vectors projects down to $g_M$ and the second decomposition is $g_P$-orthogonal, we have that 
\beq \label{NormsCurve}
- \,  g_M ( \dot{\gamma}_M, \, \dot{\gamma}_M)\ = \  - \, g_P (\dot{\gamma}^\HH, \, \dot{\gamma}^\HH) \, = \,  g_P (\dot{\gamma}^\VV, \, \dot{\gamma}^\VV) \; - \; g_P (\dot{\gamma}, \, \dot{\gamma}) \ .
\eeq
The restriction of $g_P$ to vertical vectors is the Riemannian metric $g_K$ on the fibre. Hence, the first term on the right-hand side is always non-negative. So is the second term when $\gamma(s)$ is a null or timelike curve on $P$. In that case, the projection $\gamma_M (s)$ is also a null or timelike curve on $M_4$. 
The particle's proper time on Minkowski space along the path $\gamma_M (s)$ is measured by
\beq \label{DifferenceProperTime}
c\, [ \tau(s_1)\, -\, \tau(s_2) ] \ = \ \int_{s_1}^{s_2} \sqrt{-  g_M\left(\dot{\gamma}_M , \dot{\gamma}_M  \right) } \ \dd s \ = \  \int_{s_1}^{s_2} \sqrt{g_K (\dot{\gamma}^\VV,  \dot{\gamma}^\VV) \, - \, g_P (\dot{\gamma},  \dot{\gamma}) }  \; \dd s \ .
\eeq
The last integral depends on $g_K$ and on the gauge fields, as is clear from \eqref{DecompositionsTangentCurve}.

\section{Geodesics in Riemannian submersions}
\label{GeneralGeodesicsSubmersions}

\subsection*{Equations of motion}

This section describes the main mathematical results in the paper. Some readers may wish to glance through the details in a first reading and only extract the main outputs, such as formulae \eqref{SecondDerivativeProperTime}, \eqref{IdentitySecondFundForm}, \eqref{GeneralizedLorentzForce2} and propositions \ref{HorizontalGeodesics} and \ref{EquivalenceKillingCondition}.

Let $\gamma(s)$ be a general curve on $P$ and let $\nabla$ denote the Levi-Civita connection on $TP$ associated with the metric $g_P$. The covariant derivative $\nabla_{\dot{\gamma}} \, \dot{\gamma}$ determines the parallel transport of the tangent vector $\dot{\gamma}$ along the path $\gamma(s)$.
When $P = M_4 \times K$ and the metric $g_P$ defines a Riemannian submersion, one can ask how $\nabla_{\dot{\gamma}} \, \dot{\gamma}$ decomposes into horizontal and vertical parts.  Adapting the notation and using the properties of the tensors involved, as in \cite{Bap}, classic results of O'Neill \cite[corollary 1]{ONeill} imply that
\bal \label{DecompositionGeodesicEquation}
g_P \big(\nabla_{\dot{\gamma}} \dot{\gamma} , \, V \big) \ &= \   g_P \big( \nabla_{\dot{\gamma}} \, \dot{\gamma}^\VV , \, V \big)    \, - \,   g_P \big( S_V \, \dot{\gamma}^\VV, \, \dot{\gamma}^\HH \big)       \linebr
g_P \big(\nabla_{\dot{\gamma}} \dot{\gamma} , \, Z \big) \ &= \   g_M \big( \nabla^M_{\dot{\gamma}_M} \, \dot{\gamma}_M ,\, \pi_\ast Z \big)   \, +\,  F_A^a(\pi_\ast Z,\, \dot{\gamma}_M )\ g_P \big(  e_a ,  \,  \dot{\gamma} \big) \, +\, g_P \big( S_{\dot{\gamma}^\VV} \, \dot{\gamma}^\VV, \, Z \big)     \nonumber \ 
\end{align}
for any curve $\gamma (s)$ on $P$, any vertical vector $V$ and any horizontal vector $Z$ in $TP$. The first equation determines the vertical component of $\nabla_{\dot{\gamma}} \dot{\gamma}$; the second its horizontal part. The notation follows section \ref{NullCurves} and \cite[sec. 2]{Bap}. It differs from the conventional notation in the literature about Riemannian submersions since the latter clashes with established physics notation (see the remark in appendix \ref{Auxiliary results}). So we use the decomposition $TP = \HH \oplus \VV$;  the curvature $F_A$ is a two-form on $M_4$ with values on the vector fields on $K$; the tensor $S: \VV \times \VV \to \HH$ is the second fundamental form of the fibres of the projection $\pi: P \rightarrow M_4$, which can be idientified with the covariant derivative $\dd^A g_K$ of \eqref{DefinitionCovariantDerivative} through formula \eqref{IdentitySecondFundForm}; the symbol $\nabla^M$ denotes the Levi-Civita connection on $(M_4, g_M)$. The covariant derivative $\nabla^M_{\dot{\gamma}_M} \, \dot{\gamma}_M $ is a vector field along the curve $\gamma_M(s)$ on $M_4$. It does not vanish in general, since the projected curve $\gamma_M = \pi \circ \gamma$ need not be a geodesic on the base $(M_4, g_M)$.

Now let $\gamma(s)$ be a geodesic curve on $P$ satisfying $\nabla_{\dot{\gamma}} \dot{\gamma} = 0$. The first equation in \eqref{DecompositionGeodesicEquation} implies that the norm of the vertical component $\dot{\gamma}^\VV$ evolves according to
\beq   \label{NormVerticalComponent}
\frac{\dd}{\dd s} \ g_P \big(\dot{\gamma}^\VV, \, \dot{\gamma}^\VV   \big) \ = \ 2\ g_P \big(  \nabla_{\dot{\gamma}} \, \dot{\gamma}^\VV, \, \dot{\gamma}^\VV   \big) \ = \ 2\ g_P \big( S_{ \dot{\gamma}^\VV } \, \dot{\gamma}^\VV , \,  \dot{\gamma}^\HH \big) \ .
\eeq
Using \eqref{NormsCurve} and the fact that $g_P (\dot{\gamma}, \, \dot{\gamma})$ is constant along a geodesic, this implies that 
\beq \label{SecondDerivativeProperTime}
\frac{\dd}{ \dd s}\  \Big( c\, \frac{\dd \tau}{ \dd s} \Big)^2  \ = \ - \, \frac{\dd}{\dd s} \ g_M \big(\dot{\gamma}_M, \, \dot{\gamma}_M  \big) \ = \ 2\ g_P \big( S_{ \dot{\gamma}^\VV } \, \dot{\gamma}^\VV , \,  \dot{\gamma}^\HH \big) \ .
\eeq
So the norm on $M_4$ of the tangent $\dot{\gamma}_M$ may not be constant as the parameter $s$ varies. Note that the norm of the full tangent vector $\dot{\gamma}$ is always preserved along the geodesic. It is only the norm of the components $\dot{\gamma}^\VV$ and $\dot{\gamma}^\HH$ that may change, with opposite signs, in regions where the tensor $S$ is non-zero, i.e. in regions where the fibres of $P$ are not totally geodesic. 

Now consider the second equation in \eqref{DecompositionGeodesicEquation}. If $\gamma(s)$ is a geodesic on $P$, it implies that the projected curve $\gamma_M = \pi \circ \gamma$ on Minkowski space satisfies
\beq  \label{GeneralizedLorentzForce1}
g_M \big(\nabla^M_{\dot{\gamma}_M}\, \dot{\gamma}_M , \, \pi_\ast Z \big) \ = \ - \ g_P (e_a, \, \dot{\gamma}) \, F_{A}^a \big(\pi_\ast Z, \, \dot{\gamma}_M \big) \ - \ g_P \big(S_{\dot{\gamma}^\VV} \, \dot{\gamma}^\VV, \, Z \big)
\eeq
for all horizontal vectors $Z \in TP$. The restriction of $g_P$ to the fibres, denoted $g_K$, is a family of Riemannian metrics on $K$ parameterized by the points in $M_4$. The second fundamental form $S$ is equivalent to the covariant derivative of $g_K$ along vectors on $M_4$, as defined in \eqref{DefinitionCovariantDerivative}. This follows from the identity 
\beq \label{IdentitySecondFundForm}
(\dd^A g_K)_{\pi_\ast Z} (U, V) \ = \ - \, 2\; g_P (S_U V,  Z) \ ,
\eeq
justified in \cite[sec. 2]{Bap}. Therefore, denoting by $X$ the projection $\pi_\ast Z$ in $TM$ and using that $e_a$ is a vertical vector, the equation above can also be written as
\beq  \label{GeneralizedLorentzForce2}
g_M \big(\nabla^M_{\dot{\gamma}_M}\, \dot{\gamma}_M , \, X \big) \ = \ g_K (e_a, \, \dot{\gamma}^\VV) \, F_{A}^a ( \dot{\gamma}_M , X ) \ + \ \frac{1}{2} \; (\dd^A g_K)
_X (\dot{\gamma}^\VV, \dot{\gamma}^\VV)
\eeq
for any $X$ in $TM$. This equation generalizes the Lorentz force equation of motion to regions of spacetime where $S \neq 0$, i.e. to regions where the internal metric $g_K$ is not covariantly constant as one moves along $M_4$. It generalizes the equation derived by Kerner when $K$ is a Lie group and $g_K$ is a bi-invariant metric \cite{Kerner}, and also the equations for geodesic motion when $K$ is a circle or a group manifold with variable size \cite{Kovacs, GK, Feher, CP, LM}. The motion $\gamma_M (s)$ on $M_4$ is coupled to the internal motion $\gamma_K (s)$ through the combination $\dot{\gamma}^\VV =  \dot{\gamma}_K   \,  - \,  A^a (\dot{\gamma}_M) \, e_a$, described in \eqref{DecompositionsTangentCurve}.

\subsection*{Horizontal geodesics}

Equations \eqref{DecompositionGeodesicEquation} imply that a horizontal curve $\gamma^{\rm hor}(s)$ on $P$, i.e. a curve whose tangent has a vanishing component $(\dot{\gamma}^{\rm hor})^\VV$,  satisfies the relation
\[
g_P \big(\nabla_{\dot{\gamma}^{\rm hor}} \, \dot{\gamma}^{\rm hor} , \, E \big) \ = \   g_M \big(\,  \nabla^M_{\dot{\gamma}_M^{\rm hor}} \, \dot{\gamma}_M^{\rm hor} ,\; \pi_\ast E \, \big)     
\]
for every tangent vector $E$ in $TP$. So the curve $\gamma^{\rm hor}$ is a geodesic on $P$ if and only if its projection $\gamma_M^{\rm hor}$ is a geodesic on $M_4$. Thus, a  direct consequence of \eqref{DecompositionGeodesicEquation} is that
\begin{proposition}[\! \cite{ONeill}]  \label{HorizontalGeodesics}
The horizontal geodesics on $(P, g_P)$ are exactly the horizontal lifts of geodesics on the base $(M_4, g_M)$.
\end{proposition}
In particular, if the gauge fields $A^a$ and the internal metrics $g_K$ change but the metric $g_M$ on the base does not, then the horizontal geodesics will change as curves on $P$, but their projection to $M_4$ will remain identical.

Now suppose that $\gamma(s)$ is a causal (timelike or null) curve on $P$ that projects down to a null curve on $M_4$. This means that the particle represented by $\gamma$ is moving at the speed of light on Minkowski space. Using that $g_K (\dot{\gamma}^\VV, \, \dot{\gamma}^\VV)$ and $- g_P (\dot{\gamma}, \, \dot{\gamma})$ are both non-negative, it follows from \eqref{NormsCurve} that a zero $g_M (\dot{\gamma}_M, \, \dot{\gamma}_M)$ implies that those two terms must be zero as well.  In other words, $\gamma$ projects down to a null curve on $M_4$ if and only if $\gamma$ itself is horizontal and null on $P$. Thus, we conclude that particles moving at the speed of light on $M_4$ are always represented by horizontal, null geodesics on $P$. Combining this with the previous observations about horizontal geodesics, we get that:
\begin{consequence}
In a causal, higher-dimensional geodesic model, particles moving at the speed of light on $M_4$ are always represented by horizontal, null geodesics on $P$. Their 4D motion is not affected by the configuration of gauge fields and internal metrics on $P$. They follow the null geodesics on $M_4$ determined by $g_M$ alone.
 \end{consequence}
This is not true for particles travelling at lower speeds on $M_4$. In that case, the corresponding geodesics on $P$ can have a non-zero vertical component $\dot{\gamma}^\VV$ that, according to \eqref{GeneralizedLorentzForce2}, couples $\gamma_M(s)$ to the gauge fields and internal geometry through the tensors $F_A$ and $\dd^A g_K$.
As mentioned in the Introduction, this property of the geodesic model is incompatible with the motion of the old massless neutrinos, since they were thought to interact with the weak field and travel at the speed of light on Minkowski space. So a causal, geodesic model on $M_4 \times K$ disfavours the existence of massless neutrinos or similar particles. This does not seem to have been remarked before.

\subsection*{Vertical Killing fields}

In the next few paragraphs, we will study the conditions necessary for a Killing field of $g_K$ to be a Killing field of $g_P \simeq (g_M,  A, g_K)$ as well. When this happens, we get additional constants of the higher-dimensional, geodesic motion. 

Let $\gamma (s)$ be a geodesic for $g_P$ and let $V$ be a vertical vector field on $P$. In general, the inner-product $g_P (V,  \dot{\gamma})$ is not constant along the geodesic. Using lemma \ref{Lemma1Appendix} in the appendix, its dependence on the parameter $s$ is calculated to be
\begin{align}  \label{DerivativeChargeGeodesic}
2\; \frac{\dd}{\dd s} \ g_P (  V,  \dot{\gamma} ) \, &= \, 2\, g_P \big( \nabla_{\dot{\gamma}} \,  V, \, \dot{\gamma} \big) \ = \  (\Lie_V g_P)(\dot{\gamma}, \dot{\gamma})   \linebr
&= \, (\Lie_V g_P)\big(\dot{\gamma}^\HH, \dot{\gamma}^\HH \big)  \, + \,  2 \, \big(\Lie_V g_P)(\dot{\gamma}^\HH, \dot{\gamma}^\VV \big)  \,+ \, (\Lie_V g_P)\big(\dot{\gamma}^\VV, \dot{\gamma}^\VV \big)   \nonumber  \linebr
&= \ (\Lie_V g_K) \big(\dot{\gamma}^\VV, \dot{\gamma}^\VV \big) \, + \,  2\, g_K \big([\dot{\gamma}^\HH , V], \dot{\gamma}^\VV \big) \nonumber  \linebr
&= \ (\Lie_V g_K) \big(\dot{\gamma}^\VV, \dot{\gamma}^\VV\big) \, + \,  2\,  g_K \big(  \dd^A_{\dot{\gamma}_M}\, V   , \  \dot{\gamma}^\VV \big)   \ . \nonumber 
\end{align}
In the last equality, we used the $\Diff (K)$-covariant derivative of a vertical field $V$ on $P$ along a vector field $X$ on $M$. Using the Lie bracket of vector fields on $P$, it is defined through the expression
\beq \label{CovariantDerivativeVerticalField}
  \dd^A_X\, V  \ : = \  [X^\HH, V] \ = \ [X, V] \ + \ [A(X), V] \ = \ (\dd V^j) (X) \, \frac{\partial}{\partial  y^j}  \, + \, A^a(X)\; [e_a, V]    \ ,
\eeq
where the $y^j$ are any coordinates on $K$. So $\dd^A_X\, V$ is another vertical field on $P$.  These equalities use the fact that $V$ is vertical; that the functions $A^a(X)$ do not depend on the coordinates $y^j$; and that the Lie bracket $[ X, \, \partial / \partial y^j ]$ vanishes, since these are vector fields on different manifolds. The covariant derivative $\dd^A$ is an interesting object of study. As in \cite[sec. 2.5]{Bap}, one can check that it is $C^\infty(M)$-linear in both entries and is equivariant with respect to $\Diff (K)$-gauge transformations of $V$ and the gauge one-form $A$. 

Now, at a fixed point on $P$ there are geodesics passing through with arbitrary vectors $\dot{\gamma}^\VV$ and $\dot{\gamma}_M$. So formula \eqref{DerivativeChargeGeodesic} implies the equivalence of conditions 2 and 3 below.
\begin{proposition}  \label{EquivalenceKillingCondition}
 Let $V$ be a vertical vector field on $P$ and let $g_P \simeq (g_M, A, g_K)$ be a submersive metric. Then the following conditions are equivalent:
\begin{enumerate}
\item $V$ is a Killing vector field of the higher-dimensional metric $g_P$;
\item The restriction of $V$ to each fibre is Killing and $\dd^A_X\, V$ vanishes for all $X \in TM$;
\item The inner-products $g_P (  V,  \dot{\gamma})$ are constant for all geodesics $\gamma(s)$ on $P$.
\end{enumerate}
\end{proposition}
The equivalence of conditions 1 and 2 follows from lemma \ref{Lemma1Appendix} in the appendix.

\section{Constants of motion among massless gauge fields}
\label{MasslessGaugeFields}

\subsection*{Simplified geodesic equation}

The aim of this section is to identify constants of geodesic motion in certain regions of spacetime, namely, regions where the Higgs-like scalars are constant and where all massive gauge fields vanish. In the physical world, this would allow an electromagnetic field but would exclude a non-zero weak field, for example. Under these conditions, the classical mass and charge of physical particles are constant. So the definitions of mass and charge adopted in the geodesic model should be searched among quantities that are constants of motion in these regions.

Consider a higher-dimensional submersive metric characterized by the equivalent data $ g_P \simeq (g_M , A, g_K)$, as before. In this section, we will assume :
\begin{itemize}
\item[{\bf H1)}]  The internal metrics $g_K$ are the same for all fibres;
\item[{\bf H2)}] The one-form $A(X)$ has values in the space of Killing vector fields on $(K, g_K)$.
\end{itemize}
According to \cite{Bap} and the mass formula \eqref{MassFormula}, these assumptions correspond to regions of $M_4$ where the Higgs-like scalars are constant and only massless gauge fields are present. In particular, the second fundamental form of the fibres (denoted $S$ in section \ref{GeneralGeodesicsSubmersions}), which is equivalent to the covariant derivative $\dd^A g_K$ of \eqref{DefinitionCovariantDerivative}, vanishes. With these assumptions, the equations of section \ref{GeneralGeodesicsSubmersions} for higher-dimensional geodesics simplify considerably. 

For example, since the norm $g_P (\dot{\gamma}, \, \dot{\gamma}  )$ is always a constant of geodesic motion, equations  \eqref{NormsCurve} and \eqref{NormVerticalComponent} imply that, for geodesics, 
\beq   \label{PropertimeChange}
\frac{\dd}{ \dd s}\ \Big( c\, \frac{\dd \tau}{ \dd s} \Big)^2  \ = \ - \,\frac{\dd}{\dd s} \  g_M (\dot{\gamma}_M, \, \dot{\gamma}_M  ) \\ = \ \frac{\dd}{\dd s} \ g_P \big(\dot{\gamma}^\VV, \, \dot{\gamma}^\VV   \big) \ = 0 \ .
\eeq
So, in these regions, the rate of change of proper time, $\frac{\dd \tau}{ \dd s}$, is a constant of motion. In section \ref{ParticlesPhysics} we will relate it to the mass of the test particle. Moreover, since $\dd^A g_K$ vanishes and $\gamma$ is a geodesic by assumption, the two equations in \eqref{DecompositionGeodesicEquation} are simplified to
 \bal \label{DecompositionGeodesicEquation2}
g_P( \nabla_{\dot{\gamma}} \, \dot{\gamma}^\VV , \, V) \ & = \  0 \    \linebr
g_M \big( \nabla^M_{\dot{\gamma}_M} \, \dot{\gamma}_M ,\, X \big)   \ & = \    g_P (e_a, \, \dot{\gamma}) \;  F_{A}^a (\dot{\gamma}_M, X )      \nonumber \ .
\end{align}
Here $V$ is any vertical vector in $TP$ and $X$ is any vector in $TM$. The first equation describes the vertical component $\dot{\gamma}^\VV$ of the tangent $\dot{\gamma}$. It says that, although the vectors $\dot{\gamma}^\VV$ are not parallelly transported along the geodesic, as the $\dot{\gamma}$ are, at least the vertical part of the covariant derivative $\nabla_{\dot{\gamma}} \, \dot{\gamma}^\VV$ vanishes. The second equation says that the projection of the geodesic to $M_4$ is a curve $\gamma_M (s)$ satisfying something similar to a Lorentz force law \cite[sec. 4.3]{Wald}, but with more gauge fields involved. It reduces to the equation derived by Kerner when $K$ is a Lie group and $g_K$ is a bi-invariant metric \cite{Kerner}. The inner-products $g_P (e_a, \, \dot{\gamma})$ play the role of ``charges'', coupling the 4D motion $\gamma_M (s)$ to the curvature $F^a_A$ of the background gauge fields. The next few paragraphs will investigate the extent to which these inner-products are constant along the geodesics, so that \eqref{DecompositionGeodesicEquation2} truly resembles a Lorentz force equation of motion.

\subsection*{Constants of geodesic motion}

Consider a general geodesic $\gamma(s)$ on the higher-dimensional $P$. As usual, if $Z$ is a Killing vector field with respect to $g_P$, we have that
\[
\frac{\dd}{ \dd s} \ g_P(\dot{\gamma}, \, Z) \ |_{\gamma(s)} \ = \ g_P( \nabla_{\dot{\gamma}} \dot{\gamma}, \, Z) \ + \ g_P(\dot{\gamma}, \, \nabla_{\dot{\gamma}}  Z) \ = \ 0 .
\]
So the inner-product $g_P(\dot{\gamma}, \, Z)$ is constant along the geodesic $\gamma(s)$. 

For a simple product metric, $g_P = g_M + g_K$, the Killing fields of $g_P$ are sums of Killing fields of $g_M$ and $g_K$. So all the Killing fields of $g_K$ determine constants of geodesic motion, besides those determined by the isometries of $M_4$. However, these additional constants cannot be perceived from the 4D projection of motion. A product metric has vanishing gauge fields and constant internal geometry, so the second relation in \eqref{DecompositionGeodesicEquation2} says that $\gamma_M$ is a pure geodesic of $g_M$, uncoupled to the internal metric and to internal motion.

When the gauge fields are non-zero, some of the Killing fields of $g_K$ may no longer preserve the higher-dimensional metric $g_P  \simeq (g_M,  A, g_K)$ on $M_4 \times K$. So the description of the constants of motion becomes less straightforward. It is still simple enough, however, if we assume that $g_P$ satisfies conditions H1 and H2. 
Let us start by describing how the general formulae of section \ref{GeneralGeodesicsSubmersions} are simplified under H1.
\begin{lemma}    \label{ChargeRotation}
Let $g_P \simeq (g_M,  A, g_K)$ be a submersive metric on $P$ satisfying assumption H1. Let $\gamma(s)$ be a geodesic of $g_P$. If $V$ is a Killing field of $g_K$, we have that 
\beq \label{SimplifiedEvolution}
\frac{\dd}{\dd s} \ g_P (  V,  \, \dot{\gamma} ) \ = \  A^a(\dot{\gamma}_M) \  g_P ( [ e_a , \, V ], \, \dot{\gamma} )   \  .
\eeq
 \end{lemma}
\begin{proof}
Assumption H1 says that the internal metric is the same for all fibres. So if $V$ is a Killing field of $g_K$, the term with $\Lie_V g_K$ vanishes in \eqref{DerivativeChargeGeodesic}. Moreover, if we regard $V$ as a vector field on $M_4 \times K$ that is constant along the $M_4$ direction, the Lie bracket $[X, V]$ vanishes for arbitrary vector fields $X$ on $M_4$. This simplifies \eqref{CovariantDerivativeVerticalField}. The result follows from the combination of the simplified forms of \eqref{DerivativeChargeGeodesic}  and \eqref{CovariantDerivativeVerticalField}.
\end{proof}
Now, according to assumption H2, the gauge one-form $A(X)$ has values on the space of Killing fields of $g_K$. So when $V$ commutes with all other Killing fields on $K$, formula \eqref{SimplifiedEvolution} implies that $g_P (  V,  \, \dot{\gamma} )$ has vanishing derivative with respect to the geodesic parameter $s$. Using the equivalence relations in proposition \ref{EquivalenceKillingCondition}, we conclude that:
\begin{lemma}  \label{ConstancyCharges}
Let $g_P \simeq (g_M,  A, g_K)$ be a submersive metric on $P$ satisfying assumptions H1 and H2. Let $\kil$ be a Killing vector field on $(K, g_K)$ that commutes with all other Killing fields of $g_K$. Then $\kil$ is also a Killing field on $(P, g_P)$ and the inner-product $g_P (  \kil,  \, \dot{\gamma} ) $ is a constant of motion for every geodesic $\gamma$ on $P$.
 \end{lemma}
If the vector field $\kil$ does not commute with all other Killing fields of $g_K$, then, in general, $\kil$ will not be Killing for $g_P$. However, it is still possible to extract a constant of motion from the subalgebra of the Killing algebra that contains $\kil$. To see how this comes about, we must describe the Killing algebra of $g_K$ in more detail. 

Since $K$ is a compact manifold by assumption, the isometry group of $g_K$ is a finite-dimensional, compact Lie group. The corresponding Lie algebra can be identified with the algebra $\mathfrak{k}$ of Killing vector fields on $K$. It admits a decomposition of the form
\beq  \label{DecompositionKillingAlgebra}
\mathfrak{k} \ = \  \mathfrak{a}_1 \oplus \cdots \oplus \mathfrak{a}_m \ ,
\eeq
where the $ \mathfrak{a}_r$ are either $\mathfrak{u}(1)$ lines or simple, non-abelian Lie algebras. To describe the constants of motion associated to the summands, let $\{ \kil_{r, b}: 1 \leq b \leq \dim \mathfrak{a}_r \}$ denote a basis of $\mathfrak{a}_r$. Each $\kil_{r, b}$ is a Killing field of $g_K$. Since $K$ is compact, it is possible to choose the basis so that its elements are $L^2$-orthonormal on $(K, g_K)$ and the structure constants defined by $[\kil_{r, b}, \, \kil_{r, c}] = (f_r)_{bc}^d\,  \kil_{r, d}$ are totally anti-symmetric in their three indices (see lemma \ref{BasisChoice}). With this choice, we can consider the scalar function
\beq  \label{GeneralCharge}
 \sum\nolimits_{b} \ [\, g_P(\kil_{r, b},  \,\dot{\gamma}) \, ]^{2}
\eeq
associated to the Lie subalgebra $\mathfrak{a}_r$ and the curve $\gamma (s)$ on $P$. Later, in section \ref{ChargeVariation}, we will call it the squared $\mathfrak{a}_r$-charge of the particle represented by $\gamma$, up to a constant factor. This function is independent of the choice of basis $\{ \kil_{r, b} \}$ because, by assumption, this is a $L^2$-orthonormal basis of $\mathfrak{a}_r$. And all such bases are related to each other by linear orthonormal transformations, which preserve the sum of squares \eqref{GeneralCharge}. With these definitions, we have the following result.
\begin{proposition}  \label{ConstancyCouplingConstants}
Let $g_P \simeq (g_M,  A, g_K)$ be a submersive metric on $P$ satisfying assumptions H1 and H2. For each subalgebra $\mathfrak{a}_r$ of the Killing algebra of $g_K$, the scalar \eqref{GeneralCharge} is a constant of motion for every geodesic $\gamma(s)$ on $P$. 
 \end{proposition}
\begin{proof} 
To justify the assertion, observe that formula \eqref{SimplifiedEvolution} implies that
\[
\frac{\dd}{\dd s} \ [\, g_P(\kil_{r, b},  \,\dot{\gamma}) \, ]^{2} \; = \; 2\; g_P(\kil_{r, b},  \,\dot{\gamma}) \ \frac{\dd}{\dd s} \ g_P(\kil_{r, b},  \,\dot{\gamma})  \; = \; 2\; g_P(\kil_{r, b},  \,\dot{\gamma})  \    g_P ( [ A(\dot{\gamma}_M) , \, \kil_{r, b} ], \, \dot{\gamma} )    \ .
\]
By assumption H2, the contraction $A(\dot{\gamma}_M)$ is a Killing vector field on $K$, so a vector in the Killing algebra \eqref{DecompositionKillingAlgebra}. Its component in the subspace $\mathfrak{a}_r$ to which $\kil_{r, b}$ belongs is just $ \sum_c A^c(\dot{\gamma}_M) \, \kil_{r, c}$. Since \eqref{DecompositionKillingAlgebra} is a Lie algebra decomposition, not just a vector space decomposition, this component is the only one that matters when calculating the bracket $[ A(\dot{\gamma}_M) , \, \kil_{r, b} ]$. All other components commute with $\kil_{r, b}$. So we have that
\[
[ A(\dot{\gamma}_M) , \, \kil_{r, b} ] \ = \ \sum_c  \, A^c(\dot{\gamma}_M)  \; [ \kil_{r, c} , \, \kil_{r, b} ] \ = \ \sum_{c,d} \, A^c(\dot{\gamma}_M)  \; (f_r)_{cb}^d \; \kil_{r, d} \ .
\]
Thus, we finally get that the derivative is
\beq
\frac{\dd}{\dd s} \ \sum_{b} \ [\, g_P(\kil_{r, b},  \,\dot{\gamma}) \, ]^{2} \ = \ 2\,  \sum_{b, c, d} \  g_P(\kil_{r, b},  \,\dot{\gamma})  \  A^c(\dot{\gamma}_M)  \; (f_r)_{cb}^d \  g_P(\kil_{r, d},  \,\dot{\gamma})   
\ = \ 0 \ , \nonumber
\eeq
where the last equality follows from the anti-symmetry of the structure constants $f_r$ in the indices $b$ and $d$.
\end{proof} 
This proposition gives us constants of geodesic motion in regions of spacetime where all gauge fields are massless and the internal metric is constant. However, inspecting the calculation in the proof, one quickly realizes that some of the original assumptions may be relaxed. We can drop H2 and accept non-zero massive gauge fields as long as those fields commute with the Killing fields of $g_K$. More precisely, the previous proposition can be generalized by a similar calculation to the following result.
\begin{proposition}  \label{ConstancyCouplingConstants2}
Let $g_P \simeq (g_M,  A, g_K)$ be a submersive metric on $P$ satisfying assumption H1. Suppose that the gauge one-form $A(X)$ has values in a subspace of vector fields on $K$ of the form $\mathfrak{a}_r \oplus \mathfrak{h}$ with the commutation relation $[\mathfrak{a}_r, \mathfrak{h}]= 0$. Then the scalar \eqref{GeneralCharge} is a constant of motion for every geodesic $\gamma(s)$ on $P$. 
 \end{proposition}
The difference between the two propositions can be illustrated in the case where $K=G$ is a compact, simple Lie group and $g_K$ is a generic left-invariant metric. For more details about these metrics, see \cite{Milnor, Bap}, for example. The Killing algebra of $g_K$ is the space of right-invariant vector fields on $G$, denoted $\mathfrak{g^\RR}$. If $\{\kil_{b}\}$ is a basis of $\rm{Lie}(G)$ that is orthonormal with respect to the Killing form, then its right-invariant extension to $G$, $\{\kil_{b}^\RR\}$, is a basis of $\mathfrak{g}^\RR$ with totally anti-symmetric structure constants. The left-invariant vector fields on $G$ are not Killing, but their Lie bracket with the right-invariant fields does vanish. So we are in the conditions of Proposition \ref{ConstancyCouplingConstants2}, with $\mathfrak{a}_r = \mathfrak{g^\RR}$ and $\mathfrak{h} = \mathfrak{g^\LL}$. This means that the scalar $\sum\nolimits_{b} \, [\, g_P(\kil_b^\RR,  \,\dot{\gamma}) \, ]^{2}$ will be a constant of motion even in the presence of the massive gauge fields associated to the left-invariant vector fields on $K$.

A similar example shows that these constants of geodesic motion are not necessarily all independent of each other. To this end, still take $K$ to be a compact, simple Lie group, but now let $g_K$ be a bi-invariant metric. Then the Killing algebra of $g_K$ is the sum of the spaces of left-invariant and right-invariant vector fields, $\mathfrak{g}^\LL \oplus \mathfrak{g^\RR}$.  If $\{\kil_{b}\}$ is a basis of $\rm{Lie}(G)$ orthonormal with respect to the Killing form, then its left-invariant extension, denoted $\{\kil_{b}^\LL\}$, is a basis of $\mathfrak{g}^\LL$ with totally anti-symmetric structure constants. Moreover, at every point of $G$, the vector fields $\kil_{b}^\LL$ are $g_K$-orthogonal to each other and span the tangent space to the point. Similarly for the right-invariant extensions $\{\kil_{b}^\RR\}$. Then it is straightforward to check that the two constants of geodesic motion associated by proposition \ref{ConstancyCouplingConstants} to the summands $\mathfrak{g}^\LL$ and $\mathfrak{g}^\RR$ coincide with each other for all geodesics. Both are equal to the norm $g_K(\dot{\gamma}^\VV, \dot{\gamma}^\VV)$, up to a constant factor.

\section{Higher-dimensional momentum and mass}
\label{ParticlesPhysics}

After the mostly geometrical content of sections \ref{GeneralGeodesicsSubmersions} and \ref{MasslessGaugeFields}, in this section we come back to physics.
The motion of test particles on the higher-dimensional space $P$ is parameterized by timelike or null geodesics. Let $\gamma (s)$ be such a geodesic. The higher-dimensional momentum vector of the associated particle is the vector tangent to $P$ defined simply by 
\beq \label{DefinitionHDMomentum}
p(s)  \  := \  \sigma \  \frac{\dd \gamma}{\dd s}  \ =  \ \sigma \ \dot{\gamma} (s) \ .
\eeq
Here $\sigma$ is a positive constant with dimension $\MT$. Since $\gamma$ is geodesic, the momentum vectors are covariantly constant along the curve, $\nabla_{\dot{\gamma}}\, p = 0$. The higher-dimensional momentum vector can be decomposed in two different ways:
\beq 
p(s) \ =\  p^\HH \ + \ p^\VV  \ =\  p_M \ +\  p_K \ .
\eeq 
As in \eqref{TwoDecompositions} and \eqref{DecompositionsTangentCurve}, writing the curve $\gamma (s)$ on $M_4 \times K$ as $(\, \gamma_M (s), \gamma_K(s) \, )$, the different components of momentum are given by
\begin{align}  \label{MomentumComponents}
p_M  \ &=  \  \sigma \ \dot{\gamma}_M \ \linebr
p_K  \ &=  \  \sigma \ \dot{\gamma}_K \   \nonumber  \linebr
p^\HH  \ &=  \ \sigma \ \dot{\gamma}^\HH \ = \ p_M   \,  + \,  A^a  (p_M) \, e_a  \ \nonumber  \linebr
p^\VV  \ &=  \  \sigma \ \dot{\gamma}^\VV \ = \ p_K   \,  - \,  A^a  (p_M) \, e_a  \ \nonumber \ .
\end{align}
For each value of $s$ the vectors $p_M$ and $p_K$ can be regarded as tangent to $M_4$ and $K$, but they are not $g_P$-orthogonal to each other. The vectors  $p^\HH$ and $p^\VV$ are orthogonal on $P$, but they depend on the gauge fields and mix the components along $M_4$ and $K$.

Since $\gamma$ is timelike or null, by \eqref{NormsCurve} the same is true for its projection $\gamma_M$ to $M_4$. So $ g_M (\dot{\gamma}_M, \, \dot{\gamma}_M )$ is non-positive. The particle's 4D proper time $\tau$ is defined, up to an additive constant, by
\beq  \label{ProperTimeDefinition}
c\, \frac{\dd \tau}{\dd s} \ = \ \sqrt{ - \, g_M (\dot{\gamma}_M, \, \dot{\gamma}_M ) } \ .
\eeq
When this derivative is non-zero, $\tau$ can also be used as a parameter for the curve. Then the higher-dimensional momentum vector can be written as  
\beq \label{MassMomentumRelation}
p(s)  \  = \ m(s)\  \frac{\dd \gamma}{\dd \tau}  \ ,
\eeq
where we have defined the particle's rest mass function 
\beq  \label{MassParticle}
m(s) \ := \ \sigma \ \frac{\dd \tau}{\dd s} \ = \  \frac{1}{c}\;  \sqrt{-  g_M\left(p_M ,\, p_M  \right) } \ = \  \frac{1}{c}\;  \sqrt{g_P\left(p^\VV ,\, p^\VV  \right) - g_P\left(p,\, p \right) } \ .
\eeq
The last equality uses \eqref{NormsCurve}. The mass function is independent of the 4D frame because it is defined in terms of proper time. It is positive for all geodesics on $P$ that project down to timelike curves $\gamma_M (s)$ on $M_4$. It vanishes only when $\gamma(s)$ projects to a null curve on $M_4$. Combining \eqref{MassMomentumRelation} with formula \eqref{MomentumComponents} for the momentum components, it is clear that in 4D we have
\[
p_M \ = \ m(s)\  \frac{\dd \gamma_M}{\dd \tau} \ .
\]
Thus, as long as $p_M$ is identified with the particle's 4-momentum on $M_4$, this relation implies that $m(s)$ is the particle's mass. In particular, given inertial coordinates $(t, x^1, x^2, x^3)$ on $M_4$, the energy $E$ and 3-momentum $\bp$ of the particle parameterized by $\gamma (s)$ should be read from the projection $p_M$ through the usual (coordinate-dependent) relation\footnote{The energy of the particle measured by an observer at rest in the frame determined by $(t, x^1, x^2, x^3)$ is  $E\ = \ -\, g_M (p_M, \, \frac{\partial}{\partial t} )$ \cite[sec. 4.3]{Wald}. So for the Minkowski metric \eqref{MinkowskiMetric} we get \eqref{4MomentumComponents}.}
\beq \label{4MomentumComponents}
p_M \ =  \  \frac{E}{c^2}\;  \frac{\partial}{\partial t}  \ + \  p^1 \;  \frac{\partial}{\partial x^1}  \ + \  p^2 \;  \frac{\partial}{\partial x^2}  \ + \  p^3 \;  \frac{\partial}{\partial x^3} \ .
\eeq

Now consider the case where $\gamma$ projects down to a null curve on $M_4$. By the discussion below proposition \ref{HorizontalGeodesics}, $\gamma$ itself must be a horizontal, null curve on $P$. Then relation \eqref{ProperTimeDefinition} says that  $\dd \tau / \dd s$ vanishes, and hence proper time cannot be used to parameterize $\gamma_M(s) = (\gamma^0(s), \gamma^n(s))$. In this case, the relation between momentum and mass expressed in \eqref{MassMomentumRelation} is no longer valid, but the original definitions in \eqref{MomentumComponents} still are. Assuming that the derivative $\dot{\gamma}^0$ does not vanish, the projected momentum $p_M$ can be written as 
\beq  \label{4DMomentumVelocity}
p_M  \ =  \  \sigma \; \frac{\dd \gamma_M}{\dd s} \ = \   \sigma \;  \frac{\dd t}{\dd s} \; \frac{\dd \gamma_M}{\dd t} \ = \   \sigma \; \dot{\gamma}^0 \; \Big( \; \frac{\partial}{\partial t} \ + \  \bv \;\Big) \ ,
\eeq
with $|\bv| = c$ because  $\gamma_M(s)$ is null on $M_4$. Thus, to keep \eqref{4MomentumComponents} valid, we must relate the (massless) particle's energy and 3-momentum to the components of $p_M$ through
\beq  \label{4DEnergyMomentum}
\frac{E}{c^2} \ = \  \sigma \; \dot{\gamma}^0 \qquad \qquad \quad   p^n \ = \ \sigma\; \dot{\gamma}^0  \; v^n \ = \ \sigma\;   \frac{\dd t}{\dd s} \;  \frac{\dd \gamma^n }{\dd t}  \  = \ \sigma\; \dot{\gamma}^n \ .
\eeq
In particular, using Einstein's relation for photons, the frequency associated to the massless particle represented by the null curve $\gamma_M (s)$, in the frame $(t, x^1, x^2, x^3)$, is just $\nu = \sigma \, c^2 \, \dot{\gamma}^0 \,/ \,  h$.

Observe that, in this picture, there is no physical freedom to reparameterize affine geodesics on $P$. More precisely, if $\tilde{\gamma}(s) := \gamma (\alpha\, s)$ is a reparameterization of $\gamma$ by a positive constant, then $\tilde{\gamma}(s)$ is still a geodesic on $P$, of course, but it now represents the motion of a test particle with a different higher-dimensional momentum. Evaluating the momentum at the point $\tilde{\gamma} (s)$ in $P$, we have
\[
\tilde{p} \; |_{\tilde{\gamma} (s)} \ = \ \sigma \  \frac{\dd \tilde{\gamma}}{\dd s} \;  |_{\tilde{\gamma} (s)} \ = \ \sigma \ \alpha\ \frac{\dd \gamma}{\dd s} \;  |_{\gamma (\alpha\, s)} \ = \  \alpha \; p \; |_{\gamma (\alpha \, s)} \ =  \alpha \; p \; |_{\tilde{\gamma} (s)}\ .
\]
So $\tilde{p} = \alpha\, p$ as vector fields over the locus on $P$ of the curve.  Also, at the point $q = \tilde{\gamma} (s)$, the masses of the particles represented by $\tilde{\gamma}$ and $\gamma$, as written in \eqref{MassParticle}, are related to each other by
\[
\tilde{m}^2 \; |_q \ = \  - \, c^{-2}\; g_P\big(\tilde{p}^\HH , \tilde{p}^\HH \big) \ |_q  \ = \  - \, c^{-2}\;  \alpha^2 \, g_P \big(p^\HH , p^\HH \big) \ |_q \ = \  \alpha^2 \; m^2 \; |_q \ .
\]
So $\tilde{\gamma}$ and $\gamma$ parameterize the motion on $P$ of different physical particles. In the case where $\tilde{\gamma}$ and $\gamma$ are horizontal geodesics, both $\tilde{m}$ and $m$ vanish, of course. But a similar argument implies that the 4-momenta on $M_4$ of the respective particles are related by
\[
\tilde{p}_M \; |_{\pi (q)} \ =\  \alpha\ p_M \; |_{\pi (q)} \ .
\]
Hence, the two geodesics describe the motion of two massless particles with different energies in the frame $(t, x^1, x^2, x^3)$.

\section{Rest mass variation}
\label{MassVariation}

This section studies one of the most interesting features of Kaluza-Klein models: the possibility of rest mass variation of a test particle in free fall along a higher-dimensional geodesic.
According to \eqref{SecondDerivativeProperTime}, the scalar $m(s)$ defined in \eqref{MassParticle} is a constant of motion when the geodesic traverses regions in $P$ with totally geodesic fibres, i.e. regions where $S= \dd^A g_K = 0$. In those regions, the geodesic parameterizes the motion of a test particle with fixed rest mass. In regions where the internal metric $g_K$ is not covariantly constant, the same formula \eqref{SecondDerivativeProperTime} implies that the mass function can vary according to 
\beq
c^2\ \frac{\dd}{\dd s} \; m^2(s)  \ = \  2  \; g_P \big( S_{ p^\VV } \, p^\VV , \,  \dot{\gamma}^\HH \big) \ = \ -  \; (\dd^A g_K)_{\dot{\gamma}_M} (p^\VV, p^\VV) \ .
\eeq
The last equality uses \eqref{IdentitySecondFundForm}. At the same time, identity \eqref{NormsCurve} says that the different components of the momentum vector defined in \eqref{MomentumComponents} are related by
\beq  \label{RelationMassMomentum}
-\, g_M(p_M , p_M ) \; = \; - \, g_P(p^\HH , p^\HH ) \; = \;  g_P (p^\VV , p^\VV ) \; - \; g_P(p,p) \ .
\eeq
Since norm of the full momentum $g_P(p, p)$ is a constant of geodesic motion, it follows from \eqref{MassParticle} that
\beq
\frac{\dd}{\dd s} \; g_P (p^\VV , p^\VV )  \; = \;   -\, \frac{\dd}{\dd s} \; g_P (p^\HH , p^\HH ) \ = \  c^2\ \frac{\dd}{\dd s} \; m^2(s)   \ .
\eeq
So the norms of the horizontal and vertical components of momentum may change with opposite signs along the motion, while the norm of the full momentum $p$ remains constant. The picture is that, in regions with non-vanishing $\dd^A g_K$, the higher-dimensional geometry is distorted, so forces a transfer between the horizontal and vertical components of momentum of particles in free fall along geodesics. Since the norm of the horizontal momentum corresponds in 4D to the particle's rest mass, in those regions we observe a phenomenon of rest mass variation.

For example, suppose that for $s< s_1$ a geodesic $\gamma (s)$ traverses a region of $P$ with vanishing tensor $\dd^A g_K$; after that, it enters a region where massive gauge fields are present and the internal scalars are not constant; then, for $s > s_2$ it comes out into another region with vanishing $\dd^A g_K$. The test particle parameterized by this geodesic has a well-defined and constant mass inside the first and third regions. Those two masses may not coincide, because the geodesic traversed a region with changing internal geometry, which is able to transfer between vertical momentum and horizontal momentum. The classical mass change in the overall process is given by
\[
m^2(s_2) \ - \ m^2(s_1) \ = \ -\; \frac{1}{c^2} \;  \int_{s_1}^{s_2}  (\dd^A g_K)_{\dot{\gamma}_M} (p^\VV, p^\VV) \ \dd s  \ .
\]
Qualitatively, this property of geodesics agrees with the physical fact that in regions where massive gauge fields are present, or where massive scalar fields (such as the Higgs field) are non-constant, particles may interact with the respective massive bosons and undergo a process of mass change. Of course, in this case, we would not call the incoming and outgoing objects the same particle anymore, since, by convention, particles with different rest masses are called different names.

\section{Charges and charge variation}
\label{ChargeVariation}

After mass, we consider charge. 
Let $\kil$ be a Killing vector field on $(K, g_K)$ that commutes with all other Killing fields of $g_K$. Let $\gamma (s)$ be a geodesic on $P$ parameterizing the motion of a test particle with momentum $p = \sigma \, \dot{\gamma}$. We define the particle's charge with respect to $\kil$ to be the scalar 
\beq   \label{DefinitionCharge}
q_\kil (s) \ := \ -\, g_P(\kil, \, p )  \ = \ -\, g_K(\kil, \, p^\VV ) \ .
\eeq
So $q_\kil (s)$ measures the component of the particle's momentum along the vertical Killing field $\kil$. According to lemma \ref{ConstancyCharges}, this scalar is a constant of motion when the geodesic traverses regions of $P$ that satisfy conditions H1 and H2. Just like the particle's mass $m(s)$. Now suppose that in those regions the gauge form has values in the span of $\kil$, i.e. suppose that we can write $A(X) = A^\kil(X)\, \kil$ for all $X$ tangent to $M_4$. Then formula \eqref{DecompositionGeodesicEquation2} says that the projection to $M_4$ of the geodesic $\gamma (s)$ satisfies  
\[
g_M \big( \nabla^M_{\dot{\gamma}_M} \, \dot{\gamma}_M ,\, Y \big)   \  = \   \frac{1}{\sigma} \ q_\kil(s) \  F_{A}^\kil \big(Y, \, \dot{\gamma}_M \big) \ .
\]
But by definition of $m(s)$ we have
\[
\dot{\gamma}_M \ = \ \frac{\dd \gamma_M}{\dd s} \ = \ \frac{m(s)}{\sigma} \ \frac{\dd \gamma_M}{\dd \tau}  \ .
\]
Thus, in terms of the 4-velocity $\dd \gamma_M / \dd \tau$ on Minkowski space, we get
\beq \label{LorentzForceEquation}
g_M \Big( \nabla^M_{\frac{\dd \gamma_\MMM}{\dd \tau} } \, \frac{\dd \gamma_M}{\dd \tau} ,\ Y \Big)   \  = \  \frac{q_\kil}{m} \ \,  F_{A}^\kil \Big(Y, \ \frac{\dd \gamma_M}{\dd \tau} \Big) \ .
\eeq
So the projection to $M_4$ of the particle's motion satisfies a Lorentz force equation \cite[sec. 4.3]{Wald}, as long as we identify the scalar $q_\kil(s)$ with the particle's charge. This calculation justifies that interpretation of the conserved scalar. It is very familiar from 5D Kaluza-Klein \cite{Kaluza, Kovacs, GK, CE, KMMH}. So when the electromagnetic field is the only gauge field present, higher-dimensional Kaluza-Klein is similar to the 5D version. It is not quite the same though, and this can be used to mitigate some of the traditional difficulties of the 5D setting, as will be discussed in section \ref{DiscussionDifficulties}. 

A curious consequence of the definitions of charge and mass of a test particle is that
\beq   \label{MassChargeInequality}
|q_\kil |^2 \ = \  | g_K(\kil, \, p^\VV ) |^2 \ \leq \ g_K(\kil, \, \kil ) \ g_K(p^\VV, \, p^\VV )  \leq \ g_K(\kil, \, \kil ) \; m^2 \, c^2 \ .
\eeq
The last inequality uses \eqref{MassParticle} and the fact that $g_P(p,p)$ is non-positive for any timelike or null geodesic on $P$. Therefore, in this geodesic model, test particles with a given classical mass at a point $\gamma (s)$ cannot have arbitrarily strong charge. There is a maximum limit determined by the norm of the Killing field at that point. The inequality above is saturated only when the particle's momentum $p$ is null and, additionally, its vertical component is proportional to $\kil$ at the point.

Now suppose that condition H1 of section \ref{MasslessGaugeFields} is satisfied but condition H2 is not. In other words, we are in a region of $P$ where the internal metrics $g_K$ are constant but there are both massive and massless gauge fields present. Let $\kil$ denote the same Killing field of $g_K$ as before. It is a direct consequence of lemma \ref{ChargeRotation} and definition  \eqref{DefinitionCharge} that the charge of the test particle represented by a geodesic $\gamma (s)$ will evolve according to
\beq    \label{SimplifiedEvolutionCharge}
\frac{\dd}{\dd s} \; q_\kil (s) \ = \ - \;  \frac{\dd}{\dd s} \ g_P (  \kil,  \, p ) \ = \  A^a(\dot{\gamma}_M) \  g_P ( [ \kil , \, e_a], \, p )   \  .
\eeq
Now, by assumption, $\kil$ is an electromagnetic-like internal Killing field, i.e. it is a Killing field that commutes with all other Killing fields of $g_K$. So if only massless gauge fields are non-zero, the corresponding Lie brackets $[ \kil , \, e_a]$ will vanish on $K$, and charge is conserved. This is the content of lemma \ref{ConstancyCharges} of course. If there are massive gauge fields around, but the corresponding internal fields $e_a$ are such that $[ \kil , \, e_a] = 0$, then charge is still conserved. This is a special case of proposition \ref{ConstancyCouplingConstants2}. However, if we are in a region of spacetime with non-zero massive gauge fields $A^a$ such that  $[ \kil , \, e_a]$ does not vanish, then the charge $q_\kil (s)$ may no longer be a constant of geodesic motion. It will indeed vary if the particle's internal momentum $p^\VV$ is not orthogonal to $[ \kil , \, e_a]$ in the tangent space to $K$. 

These properties of geodesic motion agree with the physical fact that gauge interactions mediated by massive bosons will preserve a particle's charge if the bosons are neutral (i.e. if $[ \kil , \, e_a]$ vanishes). This is illustrated by the Z boson of the Standard Model. But the interactions will not preserve the particle's charge if the gauge bosons are charged  (i.e. if $[ \kil , \, e_a]$ is non-zero), as in the case of the two W bosons. So there is a natural physical interpretation of the phenomenon of geodesic charge variation described by \eqref{SimplifiedEvolutionCharge}.

The fact that 4D electromagnetic charge may vary along a geodesic in higher dimensions, as in \eqref{SimplifiedEvolutionCharge}, apparently has not been reported before. That may be related to the fact that massive gauge fields are usually discarded in the Kaluza-Klein literature, after the quick remark that their bosons will have masses in the Planck scale. This remark is not fully justified, in the author's view, having in mind the mass formula \eqref{MassFormula} (see the discussion in \cite{Bap}). For point particles, the general relations between mass, charges, and the direction of internal motion also do not seem to be entirely clear in the literature. That may be related to the fact that, in the thoroughly studied 5D setting, with its unidimensional internal space, it is not possible to distinguish between the direction of internal motion and the direction of the internal Killing vector field.

As in section \ref{MasslessGaugeFields}, under assumption H1 of a constant internal metric, it is possible to define a scalar function associated with the path $\gamma(s)$ and each summand $\mathfrak{a}_r$ in decomposition \eqref{DecompositionKillingAlgebra} of the Killing algebra of $g_K$. To define it, let $\{ \kil_{r, b} \}$ be a set of Killing fields  forming an $L^2$-orthonormal basis of $\mathfrak{a}_r$. Using that the higher-dimensional momentum vector is $p(s) = \sigma \dot{\gamma}$, we write
\beq \label{GeneralChargeParticle}
\alpha^2_{\mathfrak{a}_r, \gamma} (s) \ := \  \sum\nolimits_{b}  \big[\, g_P(\kil_{r, b},  \,p) \, \big]^{2} \ .
\eeq
We call this function the (squared) $\mathfrak{a}_r$-charge of the particle represented by the curve $\gamma$. It measures how orthogonal the particle's momentum vector is to the subspace of $T_{\gamma (s)} P$ spanned by the Killing fields in the summand $\mathfrak{a}_r$. It is a constant of geodesic motion if only massless gauge fields are present, as shown in proposition \ref{ConstancyCouplingConstants}. When massive gauge fields are non-zero, the scalar evolves along a geodesic $\gamma(s)$ according to
\beq    \label{SimplifiedEvolutionCouplingScalars}
\frac{\dd}{\dd s} \; \alpha^2_{\mathfrak{a}_r , \gamma} (s) \ = \  2\;  \sum\nolimits_{b} \;  g_P(\kil_{r, b},  \, p)  \    g_P \big( [ A(\dot{\gamma}_M) , \, \kil_{r, b} ], \, p \big)    \ .
\eeq
The derivation of this formula is similar to that of \eqref{SimplifiedEvolutionCharge}. It is justified in the proof of proposition \ref{ConstancyCouplingConstants2}. As shown in proposition \ref{ConstancyCouplingConstants}, after the decomposition $A = \sum_a A^a\, e_a$ of the gauge form on $M_4$, the massless components $A^a$ do not contribute to the right-hand side of \eqref{SimplifiedEvolutionCouplingScalars}. Only massive gauge fields do. Moreover, a massive gauge field $A^a$ such that the associated internal vector field satisfies $[e_a, v] = 0$ for all $v \in \mathfrak{a}_r$, will also not contribute to the right-hand side of \eqref{SimplifiedEvolutionCouplingScalars}. This is the content of proposition \ref{ConstancyCouplingConstants2}, of course. 
By analogy with the electromagnetic case, it can be called a $\mathfrak{a}_r$-neutral gauge field. A gauge field that does not satisfy the commutation condition can be called a $\mathfrak{a}_r$-charged field. Then, by definition of decomposition \eqref{DecompositionKillingAlgebra}, all massless gauge fields are $\mathfrak{a}_r$-neutral except those with values in $\mathfrak{a}_r$ itself. Those will be $\mathfrak{a}_r$-charged (resp. $\mathfrak{a}_r$-neutral) when $\mathfrak{a}_r$ is a non-abelian (resp. an abelian) subalgebra of the Killing algebra of $K$. Massive gauge fields on $M_4$, in turn, will generically be $\mathfrak{a}_r$-charged, but some can be $\mathfrak{a}_r$-neutral.

To end this section, we note that the charge variation described by \eqref{SimplifiedEvolutionCharge} and \eqref{SimplifiedEvolutionCouplingScalars} is different from the rotation of non-abelian isospin along a geodesic, found by Kerner and Wong in \cite{Kerner, Wong} and reviewed in \cite{HZ}, for example. The rotation of isospin was derived in the conditions of constant internal metric $g_K$ and massless gauge fields. It is trivial for abelian charges. In contrast, the charge variation described by \eqref{SimplifiedEvolutionCharge} and \eqref{SimplifiedEvolutionCouplingScalars} is due to massive gauge fields and is non-trivial for abelian charges. 
The relation between the two effects can be summarized by saying that the constants of motion of isospin rotation become variable quantities when massive gauge fields are introduced in the picture. Then \eqref{SimplifiedEvolutionCouplingScalars} describes the variation along a geodesic of those former constants of motion.

\section{A unique speed in higher dimensions}
\label{AUniqueSpeed}

Kaluza-Klein models strive for conceptual simplicity at the classical level, before thinking about quantization. Following this motto, in this section we remark the naturalness of the hypothesis that all elementary particles travel at the speed of light in higher dimensions. It is the projection of velocities to three dimensions that appears to produce speeds in the range $[0, c]$, as observed macroscopically. 
Under this assumption, the previous geodesic model is simplified, as we will see. Both massless and massive particles satisfy a photon-like energy-momentum relation in higher dimensions, which projects down to the usual energy-momentum relation in 4D. The energy stored in 3D rest mass of classical particles becomes entirely due to the kinetic energy of internal motion. 

Let us then assume that test particles always follow null geodesics on $P$. In a Riemannian submersion, we know that tangent vectors to a null path $\gamma (s)$ satisfy
\beq \label{NormsCurveNull}
- \,  g_M ( \dot{\gamma}_M, \, \dot{\gamma}_M)\ = \  - \, g_P (\dot{\gamma}^\HH, \, \dot{\gamma}^\HH) \, = \,  g_P (\dot{\gamma}^\VV, \, \dot{\gamma}^\VV) \;  \ .
\eeq
This is just a simplification of \eqref{NormsCurve}. Since $\dot{\gamma}^\VV$ is tangent to the fibres and the restrictions of $g_P$ to those fibres are the Riemannian metrics $g_K$, the right-hand side is always non-negative. It is zero only if $\dot{\gamma}^\VV$ vanishes. So the four-dimensional projection $\gamma_M (s)$ of the null path is timelike on $M_4$ when $\gamma$ has a vertical component and is null when $\gamma$ is completely horizontal. The projection of a null path on $P$ can never be spacelike on $M_4$. Thus, higher-dimensional null paths on $P$ can describe all types of causal motion on $M_4$ and never correspond to acausal ones. In the name of simplicity, it is natural to investigate the consistency of a dynamical model entirely based on higher-dimensional null paths.

From \eqref{DifferenceProperTime}, the 4D proper time of a particle moving along a null curve on $P$ satisfies the simplified relation
\[
c\, [ \tau(s_1)\, -\, \tau(s_2) ] \  = \  \int_{s_1}^{s_2} \sqrt{g_K (\dot{\gamma}^\VV, \, \dot{\gamma}^\VV)  }  \ \dd s \ .
\]
When there are no gauge fields around, we have $\dot{\gamma}^\VV = \dot{\gamma}_K$ and proper time is just a measure of the Riemannian distance travelled by the particle in the internal space.

Regarding mass, for a null path in higher dimensions, relation \eqref{MassParticle} reduces to 
\beq  \label{MassParticleNull}
m(s) \ = \  c^{-1}\,  \sqrt{g_P\left(p^\VV ,\, p^\VV  \right)  } \ .
\eeq
So the test particle's rest mass is simply the norm of its vertical momentum. Internal motion is the sole source of rest energy. Mass vanishes for particles travelling along horizontal, null paths on $P$, which of course also project down to null paths on $M_4$. So mass vanishes if and only if the particle has speed $c$ on Minkowski space. Rest mass will vary if there is a transfer between the horizontal and vertical components of momentum. This can happen when the geometry of $P$ is sufficiently distorted in comparison to the vacuum geometry $g_M + g_K$. More precisely, when  $\dd^A g_K$ does not vanish, as described in section \ref{MassVariation}. The total momentum $p(s)$ is always covariantly conserved along a higher-dimensional geodesic.

For a null curve on $P$, relation \eqref{NormsCurveNull} implies that the momentum components satisfy
\beq \label{NormsMomentumNull}
- \,  g_M ( p_M, \, p_M)\ = \  - \, g_P (p^\HH, \, p^\HH) \, = \,  g_P (p^\VV, \, p^\VV)  \, = \, c^2\, m^2 \;  \ .
\eeq
So the horizontal and vertical momenta have the same norm but opposite signs. This follows from the vanishing of $g_P(p,p)$, of course. Now suppose that $g_M$ is the Minkowski metric \eqref{MinkowskiMetric}, choose coordinates $(t, x^1, x^2, x^3)$ on $M_4$ and decompose the 4-momentum as 
\[
p_M \ = \ \frac{E}{c^2} \ \frac{\partial}{\partial t}\ + \ \bp \ ,
\]
as in \eqref{4MomentumComponents}. Then the momentum relation \eqref{NormsMomentumNull} can be rewritten as the 4D relation
\[
E^2 \ = \ c^2 \; |\bp |^2 \ + \ m^2\, c^4   \ ,
\]
or, alternatively, as the higher-dimensional relation
\beq  \label{EnergyMomentumNull}
E^2 \ = \ c^2 \; g_P \left( \bp^\HH + p^\VV, \; \bp^\HH + p^\VV \right)  \ .
\eeq
The former is the usual 4D energy-momentum relation. The latter is similar to the energy-momentum relation of photons, but in higher dimensions. This is not surprising, because we are assuming that 
test particles follow null paths on $P$, just as photons do in 4D.  

Dividing equation \eqref{EnergyMomentumNull} by $(\sigma \,c\,  \dot{\gamma}^0)^2$ and using \eqref{4DMomentumVelocity} and \eqref{4DEnergyMomentum}, we obtain the relation for velocities in the chosen frame,
\beq  \label{VelocitiesNull}
c^2 \ = \ g_P \big(\bv^\HH + v^\VV, \; \bv^\HH + v^\VV \big) \ = \ g_M(\bv, \bv) \; + \; g_K \big(v^\VV, v^\VV \big) \ .
\eeq
Thus, for non-trivial, null geodesics on $P$, the norm of the velocity vector in $\mathbb{R}^3 \times K$ is always $c$. Test particles always travel at the speed of light in higher dimensions, as expected. The projection of velocities to three dimensions is the sole reason for the appearance of speeds in the range $[0, c]$. %This generalizes the discussion in the last appendix of \cite{Bap2}.

A final remark now. Since the very beginnings of Kaluza-Klein, it has been well known that null geodesics in higher dimensions can describe the 4D motion of massive particles. However, in the traditional 5D model, it is not tenable to advocate that all elementary particles move along null geodesics on $M_4 \times S^1$. Even forgetting about the strong and weak forces. This is because the null geodesics of a 5D metric with circle isometry cannot reproduce all the combinations of mass and electromagnetic charge observed in elementary particles. For example, null geodesics in 5D cannot naturally describe particles with zero charge and non-zero mass, such as neutrinos. Or the existence of particles with the same charge but different masses moving on a common background 5D metric. In fact, when the internal space is unidimensional, mass and charge are proportional to each other for particles following null geodesics. This is a consequence of \eqref{MassParticleNull} and \eqref{DefinitionCharge}.  For this reason, in the 5D model, it is desirable to keep timelike geodesics on $M_4 \times S^1$ as representations of physical motions. The norm of the tangent to a timelike geodesic is a useful free parameter to adjust the rest mass of a particle without affecting its charge. In a geodesic model with a higher-dimensional $K$, however, the situation is very different. Now we can have arbitrary angles between the direction of internal motion and the electromagnetic Killing field on $K$. So \eqref{MassParticleNull} and \eqref{DefinitionCharge} imply that the mass and charge of a particle following a null geodesic become independent quantities, except for inequality \eqref{MassChargeInequality}. Thus, for a higher-dimensional $K$, it does not seem necessary to use timelike geodesics on $M_4 \times K$ to model the masses and (abelian or non-abelian) charges of different particles. Null geodesics seem to suffice. The exclusive use of null geodesics in higher dimensions also fits better with the commonly stated aim of representing fermions by solutions of a single, massless, Dirac-like equation for higher-dimensional spinors.

\section{Parameterizing the space of null geodesics}
\label{SectionSpaceGeodesics}

Motivated by the model described in the last section, we will now have a closer look at null geodesics on $P$. Denote by $\NN_h$ the space of non-trivial, null geodesics $\gamma(s)$ passing through a point $h$ in $P$ at the parameter value $s=0$. The trivial geodesic we are excluding is the one with image $h$ for all $s$. Standard properties of the geodesic equation say that $\NN_h$ can be parameterized by the non-zero, null vectors $\dot{\gamma}(0)$ in the tangent space $T_h P$. Each of these vectors is parallelly transported along the geodesic it determines, so if $\dot{\gamma}$ is non-zero and null at $s=0$, it will have the same properties for every value of $s$. 

Now pick a local coordinate system $(t, x^n, y^j)$ on $M_4\times K$, where the $x^n$ and $y^j$ are the coordinates in $\mathbb{R}^3$ and $K$, respectively. Write the geodesic as 
\[
\gamma(s) \ = \ \big( \, \gamma_M(s), \, \gamma_K(s) \, \big) \ = \ \big( \, \gamma^0 (s), \, \gamma^n(s), \,\gamma^j(s) \, \big) \ .
\]
Since $\dot{\gamma} (s)$ is null and non-zero, the time component $\dot{\gamma}^0(s)$ must be non-zero for all $s$. In particular, it is always positive or always negative. So the function $s \mapsto \gamma^0(s)$ is strictly monotonous. According to the sign of $\dot{\gamma}^0$ we can divide the space of geodesics $\NN_h$ into its components  $\NN_h^+$ and $\NN_h^-$, corresponding to particles moving forward and backward in time, respectively. 

Each of these components is isomorphic to $\mathbb{R}^{3+k} \setminus \! \{0\}$, where $k$ denotes the dimension of $K$. For example, thinking of $\NN_h^+$ as a subset of $T_hP$, there is a bijection
\beq  \label{Bijection1}
\varphi_1 :\; \mathbb{R}^{3+k} \setminus \! \{0\}  \ \longrightarrow \ \NN_h^+ \qquad  \qquad  (u^1, u^2, u^3 , w^j) \ \longmapsto \ (  u^\HH , \, w  ) \ ,
\eeq
where we have defined
\begin{align}    \label{ComponentsImageBijection1}
\bu \ &:= \  \sum\nolimits_{n=1}^3 \; u^n  \, \frac{\partial}{\partial x^n}   &     w \  &:= \  \sum\nolimits_j \, w^j  \, \frac{\partial}{\partial y^j}   \nonumber \linebr
u \ &:= \   \bu \ + \ c^{-1}\, \sqrt{|\bu|^2 + g_K(w,w)} \;  \frac{\partial}{\partial t} 
\end{align}
and the horizontal lift $u^\HH$ is taken according to \eqref{DefinitionHorizontalDistribution}. By construction, $\varphi_1(u, w)$ is always a non-zero, null vector in $T_hP$.

The bijection $\varphi_1$ parameterizes all the null geodesics starting at $h$ and moving forward in time. These represent the motion on $P$ of particles with different masses. Fixing the mass of the particle means restricting to a subset $\NN_h^+(m)$ of the full space $\NN_h^+$. According to \eqref{MassParticle}, after the simplification $g_P(p,p)= 0$, the mass of a particle at $h$ is just $c^{-1}\sqrt{g_P (p^\VV, p^\VV)}$. So it essentially corresponds to the norm of the vertical vector $w$ in \eqref{ComponentsImageBijection1}. Thus, under the bijection $\varphi_1$, the space of null geodesics at $h$ of mass $m$ corresponds to the subset
\beq   \label{SetGeodesicsFixedMass1}
\NN_h^+(m) \ \simeq \ \big\{ \, (u^1, u^2, u^3 , w^j) \in  \mathbb{R}^{3+k} \setminus \! \{0\} : \  (g_K\, |_h )_{ij} \, w^i w^j = \sigma^{-2}\,  c^2 \, m^2 \, \big\}
\eeq
of the full $\mathbb{R}^{3+k} \setminus \! \{0\}$. Topologically, it is clear that $\NN_h^+(m)$ is isomorphic to $\mathbb{R}^{3} \times S^{k-1}$ for non-zero $m$ and to $\mathbb{R}^{3} \setminus \! \{0\}$ for vanishing mass. 

When the internal metric $g_K$ has isometries in the region around the point $h$, we can also consider the space of null geodesics representing particles with given charges, besides a given mass. For example, one can consider the subset $\NN_h^+(q_\kil) \subset \NN_h^+$ representing particles with electromagnetic charge $q_\kil$ with respect to a Killing field $\kil$ of $g_K$. The charge condition \eqref{DefinitionCharge} imposes the constraint 
\beq  \label{ChargeCondition}
(g_K\, |_h )_{ij} \, \kil^i \, w^j  \ = \ - \; \sigma^{-1}  \, q_\kil 
\eeq
on the coordinates $w^j$ of $\mathbb{R}^{3+k} \setminus \! \{0\}$. So the subspaces $\NN_h^+(q_\kil)$ are isomorphic to $\mathbb{R}^{2+k} \setminus \! \{0\}$. 

Fixing mass and electromagnetic charge simultaneously defines the smaller subspaces $\NN_h^+(m, q_\kil) = \NN_h^+(m) \cap \NN_h^+(q_\kil)$ inside the space of null geodesics $\NN_h^+$. Adding the mass condition  in \eqref{SetGeodesicsFixedMass1} to the charge condition \eqref{ChargeCondition}, it is clear that these spaces are isomorphic to 
\beq \label{SpaceGeodesics}
\NN_h^+(m, q_\kil) \ \simeq \ \begin{cases}
\emptyset               & \text{if }\   |q_\kil|  \; > \;   m\, c\, |\kil|   \\
\mathbb{R}^{3}     & \text{if }\  |q_\kil| \; =\;   m\, c\, |\kil|  \\   
 \mathbb{R}^{3} \times S^{k-2} & \text{if } \ |q_\kil| \ <\  m\, c\, |\kil|   \ .
\end{cases}
\eeq
Here $|\kil|$ denotes the Riemannian length $\sqrt{(g_K)_h (\kil,\kil)}$ of the Killing vector field at the point $h$. The upper bound for the charge of a test particle at the point $h$, given the value of its mass, is of course the same as in formula \eqref{MassChargeInequality} of section \ref{ChargeVariation}.

The higher-dimensional momenta of particles with the strongest possible charge, $ |q_\kil|  =  m\, c\, |\kil| $, have fewer degrees of freedom. This happens because the vertical component of momentum must be completely aligned with $\kil$. Hence, the motion of a particle of maximum charge is completely determined by an initial point in $P$ and by the 3D momentum vector at that point. For particles with weaker charges, the situation is different. One also needs to choose the components of the vertical momentum at the initial point, and there is a $S^{k-2}$ of possible choices compatible with the mass and charge of the particle.

If we fix more constants of motion of the particle besides mass and electromagnetic charge, for example, the constants described in proposition \ref{ConstancyCouplingConstants}, then the space of representative null geodesics at $h$ will be further constrained.

\subsection*{Celestial spheres}

We will now describe a second parameterization of the space of null geodesics $\NN_h^+$. It uses physical velocities instead of momenta. Its construction  relies on the fact that the time component $\dot{\gamma}^0$ is always positive for geodesics in $\NN_h^+$. Given a null geodesic, define
\beq  \label{DefinitionAux1}
v \ := \ \frac{1}{\dot{\gamma}^0} \ \dot{\gamma}   
\eeq
as a tangent vector in $T_h P$. Recalling that $\gamma^0 (s)$ is just the time coordinate of the geodesic, the projection of $v$ to the 4D tangent space $T_{\pi (h)} M_4$ can be written as
\[
\pi_\ast (v) \ = \ \Big( \frac{\dd \gamma^0}{\dd s} \Big)^{-1} \, \sum_{n=1}^4 \; \frac{\dd \gamma^n}{\dd s} \ \frac{\partial}{\partial x^n} \ = \ \frac{\partial}{\partial t} \ +  \ \sum_{n=1}^3 \; \frac{\dd \gamma^n}{\dd t} \ \frac{\partial}{\partial x^n} \ = \  \frac{\partial}{\partial t} \ + \ \bv  \ .
\]
Here $\bv$ is the derivative of position in $\mathbb{R}^3$ with respect to time, so it is the 3-dimensional velocity vector of the particle in the coordinates $(t, x^1, x^2 , x^3)$. Decomposing $v$ into its horizontal and vertical components, we then have
\[
v \ = \ v^\HH \ + \ v^\VV \ = \ \Big(  \frac{\partial}{\partial t} \ + \ \bv \Big)^\HH \, + \; v^\VV \ .
\]
The vector $\bv^\HH  +  v^\VV$ is the derivative of position in $\mathbb{R}^3 \times K$ with respect to time. So it is the higher-dimensional velocity vector of the particle represented by $\gamma (s)$, in the coordinate system.
Since $\dot{\gamma}$ and $v$ are null vectors at $h$, we have $g_P(v,v)=0$, and hence
\begin{align}
c^2 \ &= \ -\, g_M \Big( \frac{\partial}{\partial t}, \;  \frac{\partial}{\partial t} \Big) \ = \ -\, g_P\Big( \Big(  \frac{\partial}{\partial t} \Big)^\HH, \;  \Big(  \frac{\partial}{\partial t}  \Big)^\HH \Big) \nonumber \linebr
&= \ g_P \big(\bv^\HH + v^\VV, \; \bv^\HH + v^\VV \big) \ = \ g_M(\bv, \bv) \; + \; g_K \big(v^\VV, v^\VV \big) \ .
\end{align}
Here we have also used the form \eqref{MinkowskiMetric} of the Minkowski metric and the fact that $g_P$ applied to horizontal vectors coincides with $g_M$ applied to the 4D projection of those vectors. 
Thus, the velocity vectors $\bv^\HH + v^\VV$ associated with null geodesics at $h$ define a Euclidean sphere $S_c^{k+2}$ of dimension $k+2$ and radius $c$ inside the tangent space $T_hP$. Generalizing the terminology of the 4D case, this will be called the celestial sphere at $h$. 

It is clear from this discussion that the correspondence between geodesics in $\NN_h^+$ and points in the celestial sphere of velocities is surjective but not injective. A reparameterization of the geodesic produces a constant factor in $\dot{\gamma}$ that cancels out in the quotient \eqref{DefinitionAux1}, so leads to the same velocity vector $v$. Such constant factors can be explicitly controlled by the time component $\dot{\gamma}^0$ of the derivative of $\gamma(s)$ at $h$. But, according to \eqref{4DEnergyMomentum}, this component is proportional to the energy of the particle in the coordinate system, 
\beq   \label{EnergyGamma0}
E \ =\ \sigma\, c^2\; \dot{\gamma}^0 \ .
\eeq
This means that the correspondence that takes a non-zero, null geodesic at $h$ to the velocity vector plus the energy of the particle it represents, 
\[
\varphi_2 :\;  \NN_h^+ \ \longrightarrow \  S^{k+2}_c \times \mathbb{R}^+ \qquad  \qquad  \ \dot{\gamma}_h \ \longmapsto \ \big(\, \bv^\HH \, +\, v^\VV , \, E \, \big)
\]
defines a second bijection that can be used to parameterize $\NN_h^+$. Of course, topologically, $S^{k+2}_c \times \mathbb{R}^+$ is isomorphic to the parameter space $\mathbb{R}^{3+k} \setminus \! \{0\}$ used in bijection \eqref{Bijection1}. The virtue of the second parameterization is that it is expressed in terms of clear physical quantities, namely the velocity and energy of the particle represented by the null geodesic.

The next question is how the subsets of geodesics with fixed mass, previously denoted by $\NN_h^+(m)$, look like under the parameterization $\varphi_2$. In other words, how are they carved out from the full parameter space $S^{k+2}_c \times \mathbb{R}^+ $. Using definition \eqref{MassParticle} for the mass of the particle associated with a geodesic, on the one hand, and using relation \eqref{4DEnergyMomentum} for the energy associated with the particle in the frame, on the other hand, it is not difficult to recognize that under the parameterization $\varphi_2$ we have:
\beq   \label{SetGeodesicsFixedMass}
\NN_h^+(m) \ \simeq \ \Big\{ \, \big( \bv^\HH  + v^\VV , \, E  \big)  \in S^{k+2}_c \times \mathbb{R}^+  : \ E^2 \; g_K (v^\VV, v^\VV) = m^2  \, c^6\, \Big\} \ .
\eeq
Thus, when $m=0$, the mass condition is satisfied if and only if $v^\VV$ vanishes. When $m > 0$, the mass condition only allows points in the celestial sphere with non-zero $v^\VV$, and then the component $E$ in $\mathbb{R}^+$ is fully determined by the norm of $v^\VV$. So, for $m > 0$, the space $\NN_h^+(m)$ is parameterized by the sphere $S_c^{k+2}$ after the exclusion of an $S^2$ corresponding to the points with vanishing $v^\VV$. The excluded points correspond to the 4D motions at the speed of light, of course.

Combining \eqref{DefinitionAux1}, \eqref{EnergyGamma0}, and the mass condition in \eqref{SetGeodesicsFixedMass}, one can check that the inverse of the parameterization $\varphi_2$ satisfies 
\begin{align}
\varphi_2^{-1}: \  S_c^{k+2} \setminus S^2 \ &\longrightarrow \ \NN_h^+(m)  \linebr
\bv^\HH + v^\VV \ &\longmapsto \  \dot{\gamma}(0) \,=  \,  \frac{m\, c}{\sigma \, \sqrt{g_K (v^\VV, v^\VV)}} \  \Big[  \Big(  \frac{\partial}{\partial t} \, + \, \bv \Big)^\HH  + \, v^\VV  \Big] \ .  \nonumber
\end{align}

Finally, it is relevant to mention that the parameterizations $\varphi_1$ and $\varphi_2$ of $\NN_h^+$ are not canonical, in the sense that they both rely on a choice of coordinates on Minkowski space.

\section{Difficulties with geodesics on Einstein backgrounds}
\label{DiscussionDifficulties}

Throughout this paper we have studied geodesic motion on the product $P= M_4 \times K$ equipped with an arbitrary submersive metric $g_P \simeq (g_M, A, g_K)$. But what are the most appropriate background metrics? Where do they come from? A natural answer is that background metrics should be solutions of a higher-dimensional field equation. 
In Kaluza-Klein models, the most common choice is considering solutions of the Einstein equations on $P$, obtained by variation of the Einstein-Hilbert action
\begin{align} \label{HDEinsteinHilbertAction}
\mathcal{E} (\tg_P) \ &= \ \frac{1}{2\, \kappa_P} \, \int_P \, (\, R_{\tg_\PPP} \, - \,  2 \Lambda \, ) \; \vol_{\tg_\PPP} \  \linebr
&= \ \frac{1}{2\, \kappa_P} \, \int_P \, \Big[\, R_{\tg_\MMM} \, - \, 2 \Lambda \, + \, R_{\tg_\KKK} \, - \, \frac{1}{4}\, |F_A|^2 \, - \,  \frac{1}{4}\,  |\dd^A \tg_K|^2  \ + \  |\dd^A  (\vol_{\tg_\KKK})|^2 \, \Big] \,\vol_{\tg_\PPP} \, .  \nonumber
\end{align}
The last equality is just \eqref{GaugedSigmaModelAction0}. Although this is a natural choice of action, its relation with the traditional equations of 4D physics is not trivial. The description of geodesic motion on such backgrounds has a few subtleties and difficulties.

\subsection*{Dimensional reduction of the Einstein-Hilbert action}

A first difficulty is that, for a submersive metric $\tg_P$ that solves the Einstein equations on $P$, its projection $\tg_M$ to $M_4$ does not satisfy a field equation resembling the traditional 4D Einstein equation, in general. Dimensional reduction of the relevant term in $\mathcal{E} (\tg_P)$ yields
\beq
\int_P \,  R_{\tg_\MMM} \,\vol_{\tg_\PPP} \ = \  \int_M \, (\Vol_{\tg_\KKK})  \, R_{\tg_\MMM} \,  \vol_{\tg_\MMM} \ .
\eeq
So the gravity component of the 4D Lagrangian does not appear in the traditional guise $R_{\tg_\MMM} \vol_{\tg_\MMM}$, but instead appears multiplied by a scalar field that can vary across $M_4$. This means that $\tg_M$ is not the usual 4D metric of GR.

The established procedure to overcome this difficulty is to transform the 4D Lagrangian from the Jordan frame to the Einstein frame through a Weyl rescaling of $\tg_M$ \cite{Cho1990, FGN}. In practice, this means that, given a higher-dimensional metric $\tg_P \simeq (\tg_M, A, \tg_K)$, one declares that the physical metric in 4D is not $\tg_M$ --- the simple projection of $\tg_P$ to the base $ M_4$---  but, instead, is the metric $g_M$ related to it by the rescaling
\beq  \label{WeylTransformation}
 g_M \ =\  e^{2\omega} \, \tg_M    \qquad \quad {\text{with}} \quad \ \ \  e^{2\omega} \ := \  \kappa_P^{-1} \,  \kappa_M  \, \Vol_{\tg_\KKK}  \ .
\eeq
Using the standard transformation rules of the Riemannian volume form and scalar curvature under Weyl rescalings \cite{Wald, Bap}, one then calculates that
 \begin{align}  \label{DimensionalReductionWithWT}
\frac{1}{2\, \kappa_P} \, \int_P \,  R_{\tg_\MMM} \,\vol_{\tg_\PPP} \ &= \  \frac{1}{2\, \kappa_M} \, \int_M \, \frac{\kappa_M (\Vol_{\tg_\KKK}) }{\kappa_P} \, e^{-2\omega} \,  \big[ \,  R_{g_\MMM} \, + \, 6\, \DAlembert_{g_\MMM} \omega  \, - \, 6\,   |\dd  \omega |^2_{g_\MMM}  \, \big]   \,  \vol_{g_\MMM} \nonumber \linebr
&= \  \frac{1}{2\, \kappa_M} \, \int_M \, \big[ \,  R_{g_\MMM} \, - \, 6\,   |\dd  \omega |^2_{g_\MMM}  \, \big]   \,  \vol_{g_\MMM}  \ .
 \end{align}
So the kinetic term of the new metric $g_M$ appears in the precise GR form. It is accompanied by the kinetic term of the scalar field $\omega$ on $M_4$, measuring the volume of the internal space. The field equation for $g_M$ coming from the action \eqref{HDEinsteinHilbertAction} will therefore resemble the 4D Einstein equation with scalar matter and radiation.

Now consider the Yang-Mills part of $\mathcal{E} (\tg_P)$. After dimensional reduction, it will not appear with the canonical normalization of 4D physics for a general, varying internal metric $\tg_K$. When $\tg_K$ changes, its isometry group can also change, and so can the structure of massive  / massless gauge fields in the model. Thus, the structure of Yang-Mills terms in the Standard Model cannot correspond to a global property of the Kaluza-Klein model. Only to a local one, valid in regions of $M_4$ where $\tg_K$ is approximately constant and close to its vacuum value $\tg_K^\Szero$, which one assumes to be the present-time conditions. Thus, the aim after dimensional reduction is to obtain 4D Yang-Mills terms with the canonical normalization only in regions of spacetime where $\tg_K \simeq \tg_K^\Szero$. With this principle established, one observes that
 \begin{align}
\int_P \,  |F_A|^2 \ \vol_{\tg_\PPP} \ &= \ \int_M  \,  \tg_M^{\mu \nu} \, \tg_M^{\sigma \rho}\, (F^a_A)_{\mu \sigma} \;  (F^b_A)_{\nu \rho}\, \Big[ \int_K \tg_K (e_a, e_b) \, \vol_{\tg_\KKK} \,  \Big]  \, \vol_{\tg_\MMM}  \nonumber \linebr
&= \ \int_M  \,  g_M^{\mu \nu} \, g_M^{\sigma \rho}\, (F^a_A)_{\mu \sigma} \;  (F^b_A)_{\nu \rho}\, \Big[ \int_K \tg_K (e_a, e_b) \, \vol_{\tg_\KKK} \,  \Big]  \, \vol_{g_\MMM}  \nonumber \ .
 \end{align}
The last equality reflects the invariance of the 4D Yang-Mills action under Weyl rescalings of $g_M$. Therefore, in the vacuum conditions $\tg_K \simeq \tg_K^\Szero$, dimensional reduction of $\mathcal{E} (\tg_P)$ will produce terms with the canonical normalization $-\frac{1}{4} (F_A^a)^{\mu \nu} (F_A^a)_{\mu \nu}$ if and only if
\beq \label{NormalizationInternalMetric}
 \int_K \tg^\Szero_K (e_a, e_b) \, \vol_{\tg^\SSzero_\KKK} \ = \ 2\, \kappa_P \, \delta_{ab}  \ .
\eeq
This is a normalization condition for the vector fields $e_a$ on $K$. It will be important below.

\subsection*{Geodesics and rescalings of the 4D metric}

Our previous study of geodesic motion on $P$ assumed that the projection to $M_4$ of the background metric $g_P  \simeq (g_M, A, g_K)$ is the physical 4D metric. For example, section \ref{ParticlesPhysics} used that the 4D proper time of a particle following a geodesic is $c \, \frac{\dd \tau}{\dd s} = \sqrt{- g_M(\dot{\gamma}_M, \dot{\gamma}_M)}$. However, in this section we have seen that if $\tg_P$ is a solution of the Einstein equations on $P$ with varying internal volume, then its projection $\tg_M$ should not be identified with the physical 4D metric $g_M$, only with a rescaled version of it. So we have a potential problem when interpreting the geodesics of $\tg_P$. Three possible strategies to deal with it are:
\begin{itemize}
\item[1)] Restrict the action $\mathcal{E} (\tg_P)$ to metrics of fixed internal volume, satisfying the normalization condition $ \kappa_M  \, \Vol_{\tg_\KKK} = \kappa_P$. Use the solutions of the respective field equations as backgrounds for geodesics. Due to the normalization, $\tg_M$ can be identified with the physical 4D metric $g_M$ in \eqref{WeylTransformation}. The formulae in sections \ref{ParticlesPhysics} to \ref{SectionSpaceGeodesics} are unchanged.
\item[2)] Take arbitrary solutions of the Einstein equations on $P$ as backgrounds for geodesics, but rescale the identification between the physical 4D-momentum of a test particle and the projection to $M_4$ of the momentum vector on $P$.  The formulae in sections \ref{ParticlesPhysics} to \ref{SectionSpaceGeodesics} need some adjustments.
\item[3)] Declare that the physical background metrics $g_P  \simeq (g_M, A, g_K)$ should not be solutions of the Einstein equations on $P$, but, instead, should be solutions of  alternative equations that allow a direct identification of $g_M$ with the physical 4D metric. For such backgrounds, the formulae in sections \ref{ParticlesPhysics} to \ref{SectionSpaceGeodesics} will remain unchanged.
\end{itemize}
Strategy 1 is clear enough. Let us now describe the options 2 and 3 in more detail.

\subsubsection*{Option 2}

All solutions $\tg_P$ of the Einstein equations on $P$ can be background metrics. Test particles follow affine geodesics $\gamma(s)$ of these metrics. For non-normalized solutions, the physical metric on $M_4$ is related to the projection $\tg_M$ of $\tg_P$ through the rescaling
\beq
g_M \ = \ \kappa_P^{-1} \,  \kappa_M  \, (\Vol_{\tg_\KKK}) \, \tg_M \ = \ e^{2\omega} \, \tg_M  .
\eeq
The higher-dimensional momentum vector is still $\tilde{p} = \sigma \dot{\gamma}$, but the physical 4-momentum of the test particle should not be identified with $\tilde{p}_M = \sigma \dot{\gamma}_M$, the projection of $\tilde{p}$ to $M_4$. Instead, the 4-momentum $p_M$ in \eqref{4MomentumComponents} should be identified with the rescaled projection
\beq
p_M \ = \ e^{-\omega} \ \tilde{p}_M \ = \ \sigma \; e^{-\omega}  \;  \frac{\dd \tau}{\dd s} \; \frac{\dd \gamma_M}{\dd \tau} \ .
\eeq
So the 3D rest mass of the particle represented by the geodesic $\gamma (s)$ now satisfies
\begin{align}
 m(s) \; &:= \; \sigma \, e^{-\omega} \; \frac{\dd \tau}{\dd s} \ = \ \sigma \, e^{-\omega} \, {c}^{-1} \, \sqrt{- g_M (\dot{\gamma}_M, \dot{\gamma}_M)} \ = \ {c}^{-1} \,  \sqrt{- g_M (p_M, p_M)}    \nonumber  \linebr
&= \ {c}^{-1} \, \sqrt{- \tg_M (\tilde{p}_M, \tilde{p}_M)} \ = \  {c}^{-1} \, \sqrt{\tg_K (\tilde{p}^\VV,  \tilde{p}^\VV) \ - \ \tg_P(\tilde{p}, \tilde{p})} \ .
\end{align}
This generalizes \eqref{MassParticle} to the case where the physical 4D metric $g_M$ does not coincide with the projection $\tg_M$ of the metric $\tg_P$ that appears in the geodesic equation. One can check that, after the rescalings, the formula for rest mass variation remains unchanged,
\beq
c^2\ \frac{\dd}{\dd s} \; m^2(s)  \ = \ -  \; (\dd^A \tg_K)_{\dot{\gamma}_M} (\tilde{p}^\VV, \tilde{p}^\VV) \ .
\eeq
In regions where the internal geometry is constant, $\tg_K \simeq \tg_K^\Szero$, the definition of charge should now be
\beq
q_{\kil}(s) \ := \ - \,  e^{\omega_\SSzero} \; \tg_K( \kil, \tilde{p}^\VV )  \ , 
\eeq
instead of \eqref{DefinitionCharge}, in order to preserve the form \eqref{LorentzForceEquation} of the Lorentz force equation of motion. Since $ e^{2\omega_\SSzero} =  \kappa_P^{-1} \,  \kappa_M  \, \Vol_{\tg_\KKK^\SSzero}$ is constant in these regions, the charge variation formula \eqref{SimplifiedEvolutionCharge} changes only by a constant  factor $e^{\omega_\SSzero}$ on its right-hand side.

\subsubsection*{Option 3}

The simple Einstein-Hilbert functional $\mathcal{E} (g_P)$ is the most natural action to generate field equations for background metrics on $P$. However, due to the difficulties it creates upon dimensional reduction, one could consider alternative actions for submersive backgrounds  $g_P  \simeq (g_M, A, g_K)$ on $P$. A first example of alternative functional is:
\begin{align}  \label{AlternativeAction1}
\mathcal{E}' (g_P) \, &:= \ \frac{1}{2\, \kappa_P} \, \int_P \, e^{-2\omega}\, (\, R_{g_\PPP} \, - \,  2 \Lambda \, ) \; \vol_{g_\PPP} \  \linebr
&= \   \int_M \frac{1}{2\, \kappa_M}  \, ( R_{g_\MMM}  \, - \,  2 \Lambda )  \vol_{g_\MMM} \, + \,  \frac{1}{2\, \kappa_P} \int_P \big( e^{-2\omega}  R_{g_\KKK}  -  \frac{1}{4}\,  e^{-2\omega} \, |F_A|^2  - \cdots \big) \vol_{g_\PPP} \, .  \nonumber
\end{align}
Here $e^{2\omega}$ is the scalar $\kappa_P^{-1} \kappa_M \Vol_{g_\KKK}$ and we have used the second equality in \eqref{HDEinsteinHilbertAction}. As in the calculations leading to \eqref{NormalizationInternalMetric}, this functional requires the normalization condition
\beq
 \int_K g^\Szero_K (e_a, e_b) \, \vol_{g^\SSzero_\KKK} \ = \ 2\, \kappa_P \, e^{2\omega}\, \delta_{ab} \ = \ 2\, \kappa_M \,\Vol_{g^\SSzero_\KKK}\, \delta_{ab}
\eeq
to produce a canonically normalized Yang-Mills term in regions where the internal metric is close to its vacuum value, $g_K \simeq g^\Szero_K$. By construction, in the second line of \eqref{AlternativeAction1} the 4D gravity term already appears in the Einstein frame. So the projection of the background metric $g_P$ to $M_4$, denoted here $g_M$, can be directly identified with the physical 4D metric.

A second example of alternative action for background, submersive metrics $g_P  \simeq (g_M, A, g_K)$ is the deformation of the Einstein-Hilbert action $\mathcal{E}$ defined by:
\begin{align}  \label{AlternativeAction2}
\mathcal{E}'' (g_M,   A, g_K) \, &:= \; \mathcal{E} (  e^{2\omega }  \, g_M , \,   A, \,  g_K )   \ \linebr
%&= \; \frac{1}{2\, \kappa_P} \, \int_P \, \Big[\, e^{2\omega }  \, R_{g_\MMM} \, - \, \frac{1}{4}\, |F_A|^2  + \, e^{4\omega} ( R_{g_\KKK}  - \, 2 \Lambda ) \, - \,  \frac{1}{4}\,e^{2\omega }  \,  |\dd^A \tg_K|^2  \, + \, e^{2\omega }  \,  |\dd^A  (\vol_{\tg_\KKK})|^2 \, \Big] \,\vol_{g_\PPP} \, .  \nonumber  \linebr
&=  \, \frac{1}{2 \kappa_M} \int_M ( R_{g_\MMM}  -  6 \, |\dd \omega|^2 )  \vol_{g_\MMM}  +  \frac{1}{2 \kappa_P} \int_P \big( e^{4\omega}  R_{g_\KKK}  -  \frac{1}{4}\, |F_A|^2  - \cdots \big) \vol_{g_\PPP}  .  \nonumber 
\end{align}
The second equality uses \eqref{HDEinsteinHilbertAction} and \eqref{DimensionalReductionWithWT}. For this functional, \eqref{NormalizationInternalMetric} is the normalization condition necessary to obtain a canonically normalized Yang-Mills term in regions where $g_K \simeq g^\Szero_K$. Once again, the 4D gravity term already appears in the Einstein frame, so the projection $g_M $ of the background metric to $M_4$ can be directly identified with the physical 4D metric. Due to the first line in \eqref{AlternativeAction2}, the field equations for $g_P \simeq (g_M, A, g_K)$ determined by the action $\mathcal{E}''$ coincide with those obtained by substituting $\tg_M \rightarrow e^{2\omega } g_M$ and $\tg_K \rightarrow g_K$ in the Einstein equations for $\tg_P \simeq (\tg_M, A, \tg_K)$. 

Perhaps the biggest conceptual disadvantage of action \eqref{AlternativeAction2} is that it is not defined for arbitrary metrics on $P$. Only for submersive metrics that can be decomposed into triples $(g_M, A, g_K)$. So the associated Kaluza-Klein model cannot be interpreted as a simple transposition of GR to higher dimensions.

\subsection*{Electromagnetic $q/m$ ratios}

As described in section \ref{ChargeVariation}, in regions where the internal metric is constant and only an electromagnetic-like gauge field is present, the projections to $M_4$ of higher-dimensional geodesics resemble the usual Lorentz force motions of charged particles in 4D. This was the original Kaluza ``miracle" in the 5D model.

However, when the background metric comes from the Einstein-Hilbert action on $P$, the range of 4D motions projected by higher-dimensional geodesics has severe limitations. For example, in the 5D model, all causal geodesics project down to Lorentz force motions on $M_4$ with $q/m$ ratios that are much lower than the physical ratios for elementary particles \cite{GK, CE}. Spacelike geodesics on $P$ are able produce higher $q/m$ ratios \cite{DO}, but in a causal model we would prefer to avoid them.

Let us examine the problem in more detail. If we set the constant $\kappa_P$ to be $\kappa_M \, \Vol_{g_\KKK^\SSzero}$, as in \cite{GK}, the normalization condition \eqref{NormalizationInternalMetric} implies that
\beq    \label{WrongNormalizationInternalMetric}
\frac{1}{\Vol_{g_\KKK^\SSzero}} \int_K g^\Szero_K (e_a, e_b) \, \vol_{g^\SSzero_\KKK} \ = \ 2\, \kappa_M \, \delta_{ab}  \ .
\eeq
In the 5D model, with $g_K^\Szero$ being the round metric on $K= S^1$ and $e_a = \kil$ being its standard Killing field, the function $g^\Szero_K (\kil, \kil)$ is constant along the circle $K$. So the normalization condition for $\kil$ reduces to
\beq
 g^\Szero_K (\kil, \kil) \ = \ 2\, \kappa_M   \ .
\eeq
This creates a problem when comparing to the Lorentz force motions coming from the geodesic picture. For instance, formula \eqref{MassChargeInequality} says that for non-trivial geodesics on $M_4 \times S^1$ we must have 
\beq
\frac{q_\kil}{m}  \  \leq \ c \, \sqrt{g^\Szero_K (\kil, \kil)} \ = \ c \, \sqrt{2\, \kappa_M} \ ,
\eeq
while the charge-to-mass ratio of the electron is many orders of magnitude above this value. This is the low $q/m$ ratios problem of the 5D geodesic model \cite{GK, CE, DO}.

A first observation is that it is only for unidimensional fibres that \eqref{WrongNormalizationInternalMetric} actually fixes the value of $ g^\Szero_K (\kil, \kil)$. In higher dimensions, only the average of $g^\Szero_K (\kil, \kil)$ over the manifold $(K, g^\Szero_K)$ is fixed. And, in principle, that average can be significantly different from the point values that appear in the geodesic equation. So a geodesic passing through points in $P$ where the norm $g^\Szero_K (\kil, \kil)$ is higher than average can project down to a 4D Lorentz force motion with a $q/m$ ratio higher than $c  \sqrt{2\, \kappa_M}$. If the geometry of $(K, g^\Szero_K)$ is very heterogeneous --- for instance if $ g^\Szero_K$ has a singularity --- the $q/m$ ratios of the projected 4D motions can be much higher than $c  \sqrt{2\, \kappa_M}$.

A second observation is that the normalization conditions \eqref{NormalizationInternalMetric} and \eqref{WrongNormalizationInternalMetric} assume that the background metric for geodesics comes from the simple Einstein-Hilbert action on $P$. For example, for the alternative action \eqref{AlternativeAction2}, which is directly in the Einstein frame, the normalization \eqref{NormalizationInternalMetric} is still necessary but we no longer have to fix the value of $\kappa_P$ to be $\kappa_M  \, \Vol_{\tg_\KKK^\SSzero}$, in order to obtain a canonically normalized Yang-Mills term.  So \eqref{WrongNormalizationInternalMetric} is no longer valid. Thus, choosing a sufficiently big $\kappa_P$, the higher bound on the $q/m$ ratios imposed by the normalization of the Killing field $\kil$ no longer excludes the charge-to-mass ratios of elementary particles. In other words, the action \eqref{AlternativeAction2} generates backgrounds without $q/m$ ratios problems.
More generally, in a realistic model the Einstein-Hilbert action on $P$ may need higher order corrections to describe the appearance of physical particles with masses at very different scales. Such higher order terms can affect the normalization condition for the Killing fields, and the corrected condition may be less problematic than  \eqref{NormalizationInternalMetric}.

\vspace*{1.5cm}

\section*{Acknowledgements}

\vspace*{-.1cm}

It is a pleasure to thank Nick Manton and Nuno Romão for helpful comments on an earlier version of this paper.

\newpage

\begin{appendices}

\section{Auxiliary results and a remark on notation}
\label{Auxiliary results}

\begin{remark} The notation used in this paper differs significantly from the conventional notation in the literature about Riemannian submersions. The modification is necessary because the latter clashes with the traditional notation in physics. Namely, the tensor called $A$ in \cite{ONeill1, ONeill, Besse, FIP} is essentially what we call $F_A$ here, since it represents the physical gauge fields strength. More precisely, the relation is
\beq \label{TranslationTensors}
F_A (X, Y) \ = \ 2\ \tilde{A}_{X^\HH} Y^\HH \ = \ [X^\HH, \, Y^\HH]^\VV \ ,
\eeq
where $X, Y$ are vector fields on the base $M$; the symbols $X ^\HH$ and $Y^\HH$ denote their lifts as basic vector fields on $P$; we have denoted by $\tilde{A}$ the tensor called $A$ in \cite{ONeill1, ONeill, Besse, FIP}; and the last equality is a standard result in Riemannian submersions.
The tensor called $T$ in \cite{ONeill1, ONeill, Besse, FIP} is what we call here $\dd^A g_K$ (or $S$ less often). This avoids confusion with torsion or with the energy-momentum tensor. It also emphasizes its physical role as a covariant derivative of Higgs-like fields. The precise relation is
\beq \label{TranslationTensors2}
(\dd^A g_K)_X  (U, V)  \ = \  (\Lie_{X^\HH} g_P) (U, V) \ = \ -\, 2\,  g_P (T_U V, \, X^\HH) \ = \  -\, 2\, g_P (\nabla_U V, \, X^\HH  ) \ ,
\eeq
where $X$ is a vector field on $M$ and $U, V$ are vertical vector fields on $P$. The last equality is the definition of the tensor $T$ in \cite{ONeill1, ONeill, Besse, FIP} in terms of the Levi-Civita connection on $P$. The first two equalities are derived in sections 2.3 and 2.5 of \cite{Bap}.
\end{remark}

\noindent
Now the auxiliary results. Classical work of Ehresmann and Hermann \cite{Ehresmann, Hermann} implies that:
\begin{proposition}
Let $\pi: (P, g_P) \rightarrow (M, g_M)$ be a Riemannian submersion. If the metric $g_P$ is complete, then $\pi$ defines a locally trivial fibration. In this case, given a path $\gamma_M(s)$ on $M$ starting at a point $x$ and given a point $p$ in $P$ such that $\pi (p)=x$, there exist a unique path $\gamma (s)$ on $P$ starting at $p$ such that $\gamma_M = \pi \circ \gamma$ and the derivatives $\dot{\gamma}$ are always horizontal vectors on $P$. This $\gamma$ is called a horizontal lift of $\gamma_M$.
\end{proposition}
\noindent
Regarding horizontal geodesics, O'Neill has the following result:
\begin{proposition}[\!\cite{ONeill}]
Let $\pi: (P, g_P) \rightarrow (M, g_M)$ be a Riemannian submersion and let $\gamma (s)$ be a geodesic on $P$. If the derivative $\dot{\gamma}$ is a horizontal vector at some point, then it is always horizontal. Moreover, in this case the projection $\gamma_M (s) = \pi \circ \gamma (s)$ is a geodesic on $(M, g_M)$.
\end{proposition}
\noindent
Using basic properties of Riemannian submersions (e.g. see \cite[sec. 9]{Besse}), we show that:
\begin{lemma}   \label{Lemma1Appendix}
Let $\pi: P \rightarrow M$ be a Riemannian submersion with metric $g_P$ determined by the equivalent data $(g_M, g_K, A)$. Let $U$ and $V$ be vertical vector fields and let $Z$ be a basic vector field on $P$. Then
\begin{align} \label{VerticalLieDerivatives}
(\Lie_V g_P) (U, U) \ &= \ (\Lie_V g_K) (U, U) \qquad \qquad \qquad  (\Lie_V g_P) (Z, Z) \ = \ 0   \nonumber  \linebr
(\Lie_V g_P) (Z, U) \ &= \  g_K ([Z, V], U) \ .
\end{align}
In particular, $V$ is Killing with respect to $g_P$ if and only if the restriction of $V$ to each fibre is Killing with respect to $g_K$ and, additionally, the Lie bracket $[V, Z]$ vanishes for every basic vector field $Z$ on $P$.
\end{lemma}
\begin{proof}
The horizontal basic field $Z$ is $\pi$-related to a well-defined vector field $\pi_\ast Z$ on the base $M$. Since $V$ is vertical, we have $\pi_\ast [V, Z] = [\pi_\ast V, \pi_\ast Z] = 0$, and the bracket $[V, Z]$ is also vertical. So by definition of Lie derivative:
\[
(\Lie_V g_P) (Z, Z) \ = \ \Lie_V ( \, g_P(Z, Z) \, ) - 2 g_P ([V, Z], Z) \ = \ \Lie_V ( \, g_M(\pi_\ast  Z, \pi_\ast Z) \, ) \ =\  0 \ .
\]
Using that $U$ and $[V, U]$ are both vertical, so are both orthogonal to $Z$, we have
\[
(\Lie_V g_P) (Z, U) \ = \ \Lie_V ( \, g_P(Z, U) \, ) - g_P ([V, Z], U) -  g_P (Z, [V, U]) \ = \   g_P ([Z, V], U)  \ .
\] 
Finally, since $U$, $V$ and $[U,V]$ are all vertical and $g_K$ just denotes the restriction of $g_P$ to the fibres, it is clear that
\begin{align}
(\Lie_V g_P) (U, U) \ &= \ \Lie_V ( \, g_P(U, U) \, ) - 2\, g_P ([V, U], U) \linebr
&= \   \Lie_V ( \, g_K(U, U) \, ) - 2\, g_K ([V, U], U)  \ = \ (\Lie_V g_K) (U, U) \ . \nonumber
\end{align}
Now suppose that $\Lie_V g_P$ vanishes. Then by \eqref{VerticalLieDerivatives} we get that $(\Lie_V g_K) (U, U)$ vanishes for all $U$. Since $\Lie_V g_K$ is symmetric in both entries, it must also vanish. Moreover, also by \eqref{VerticalLieDerivatives}, the vanishing of $(\Lie_V g_P) (Z, U)$ for all $U$ and the non-degeneracy of $g_K$ imply that $[Z, V]$ vanishes for all basic $Z$. This confirms the ``only if'' part of the last statement, while the converse is clear.
\end{proof}
\begin{lemma}  \label{Lemma2Appendix}
Let $\pi: P \rightarrow M$ be a Riemannian submersion with metric $g_P$ determined by the equivalent data $(g_M, g_K, A)$. Let $Y$ and $Z$ be basic vector fields and let $U$ be a vertical vector field on $P$. Then
\begin{align} % \label{VerticalLieDerivatives}
(\Lie_Y g_P) (Z, Z) \ &= \ \pi^\ast \left[ (\Lie_{\pi_\ast Y}\, g_M) (\pi_\ast Z, \pi_\ast Z)   \right]  \qquad \qquad   (\Lie_Y g_P) (U, U) \ = \  (\dd^A g_K)_{\pi_\ast Y} (U, U)  \nonumber  \linebr
(\Lie_Y g_P) (Z, U) \ &= \  g_P \left(  F_A(\pi_\ast Y, \pi_\ast Z) , \, U \right) \ .
\end{align}
In particular, $Y$ is Killing with respect to $g_P$ if and only if the projection $\pi_\ast Y$ is Killing with respect to $g_M$, the covariant derivative $(\dd^A g_K)_{\pi_\ast Y}$ vanishes as a tensor on $P$ and, additionally, the contraction $F_A(\pi_\ast Y, \cdot)$ vanishes as a one-form on $M$.
\end{lemma}
\begin{proof}
By definition, basic vector fields on $P$ are horizontal and hence orthogonal to vertical vector fields. Moreover, the Lie bracket of a basic field with a vertical field is always vertical on $P$, hence orthogonal to $Z$. So we have that
\begin{align}
(\Lie_Y g_P) (Z, U) \ &= \ \Lie_Y ( \, g_P(Z, U) \, ) \, - \,  g_P ([Y, Z], U)  \,- \, g_P (Z, [Y, U]) \ \nonumber \linebr
&=\  - \,  g_P ([Y, Z], U) \ = \     g_P \left(  F_A(\pi_\ast Y, \pi_\ast Z) , \, U \right)\ ,
\end{align}
where the last equality follows from \eqref{TranslationTensors}.

A basic vector field on $P$ is $\pi$-related to a unique, well-defined vector field on the base $M$. So we have that $\pi_\ast [Y, Z] = [\pi_\ast Y, \pi_\ast Z]$ as vector fields on the base $M$. Then 
\begin{align}
(\Lie_Y g_P) (Z, Z) \ &= \ \Lie_Y ( \, g_P(Z, Z) \, ) \, - \, 2\,  g_P ([Y, Z], Z)   \nonumber \linebr
&=\  \Lie_Y \,  \pi^\ast ( \, g_M (\pi_\ast Z, \pi_\ast  Z) \, ) \, - \, 2\,  g_P ([Y, Z]^\HH , Z)  \nonumber \linebr
&=\   \pi^\ast \, \Lie_{\pi_\ast Y} ( \, g_M (\pi_\ast Z, \pi_\ast  Z) \, ) \, - \, 2\,  \,  \pi^\ast\, \left[ g_M (\pi_\ast  [Y, Z] ,\pi_\ast  Z) \right] \nonumber \linebr
&=\   \pi^\ast  \left[  (\Lie_{\pi_\ast Y} g_M ) (\pi_\ast Z, \pi_\ast  Z) \right] \ 
\end{align}
as functions on $P$, as desired. The last identity
\[
(\Lie_Y g_P) (U, U) \ = \  (\dd^A g_K)_{\pi_\ast Y} (U, U)  
\]
is tautological. It comes from the definition of covariant derivative of $g_K$, as explained in \eqref{TranslationTensors} and \cite{Bap}.
Finally, the necessary and sufficient condition for $Y$ to be Killing with respect to $g_P$ is clear from the formulae above, considering that $Z$ can be any basic field and $U$ can be any vertical field on $P$. 
\end{proof}

\noindent
The following observation, stated in \cite[sec. 9F]{Besse}, is a consequence of results in \cite{Hermann, ONeill1}.
\begin{proposition}
Let the metric $g_P \simeq (g_M, A, g_K)$ define a Riemannian submersion on $P$ with totally geodesic fibres, i.e. with $\dd^A g_K = 0$. Then the curvature form $F_A(X, Y)$ has values in the Killing vector fields of $g_K$ for every $X, Y \in TM$.
\end{proposition}

\noindent
In section \ref{MasslessGaugeFields} we use the following standard result.
\begin{lemma}
\label{BasisChoice}
Let $\{ u_j \}$ be a basis of the space of Killing vector fields on a compact Riemannian manifold $(K, g_K)$. Define the structure constants with respect to this basis by the relations $[u_i, u_j] = f^l_{ij}\,  u_l$. If the basis is orthonormal with respect to the $L^2$-inner product on $(K, g_K)$, then the $f^l_{ij}$ are totally anti-symmetric in their three indices.
\end{lemma}
\begin{proof} 
The Lie derivatives of the metric $g_K$ and volume form $\vol_{g_K}$ both vanish along Killing vector fields. Using this fact, the definition of the $L^2$-inner product of vector fields and Stokes' theorem, we can write
\begin{align}
\sum\nolimits_l \; f^l_{ij} \; \langle u_l, u_k \rangle_{L^2} \ &= \  \langle [u_i, u_j] ,  u_k \rangle_{L^2}  \ = \ \int_K  g_K ([u_i, u_j] ,  u_k) \; \vol_{g_\KKK}  \nonumber \linebr
&= \  \int_K \Big\{ \, \Lie_{u_i} \big[ g_K(u_j ,  u_k) \big] \ - \ g_K (u_j, [u_i, u_k])  \, \Big\} \; \vol_{g_\KKK} \nonumber \linebr 
&= \  \int_K \Lie_{u_i}  \big[\,  g_K(u_j ,  u_k)  \ \vol_{g_\KKK}\, \big] \ - \ \int_K g_K (u_j, [u_i, u_k])  \; \vol_{g_\KKK}   \nonumber \linebr
&= \ - \ \int_K g_K (u_j, [u_i, u_k])   \; \vol_{g_\KKK}  \ = \ -\;    \sum\nolimits_l  \;  f^l_{ik} \; \langle u_j, u_l \rangle_{L^2} \nonumber \ .
 \end{align}
Thus, when the basis $\{ u_j \}$ is $L^2$-orthonormal, we get that 
\[
 \sum\nolimits_l  f^l_{ij} \; \delta_{lk} +  f^l_{ik} \; \delta_{jl} \ = \ 0  \quad  \iff \quad  f^k_{ij} = \ - \; f^j_{ik} \ .
\]
The anti-symmetry of $f^l_{ij}$ in the lower two indices follows from its definition and is true for any basis.
\end{proof}

\section{Parallel transport in Riemannian submersions}
\label{ParallelTransport}

For reference, in this appendix we write down the equations of parallel transport of vectors along general curves $\gamma(s)$ on the higher-dimensional space $P=M\times K$. That space is assumed to be equipped with a submersive metric $g_P \simeq (g_M, A, g_K)$. Parallel transport is determined by its Levi-Civita connection $\nabla$. As in section \ref{NullCurves}, a vector $E$ in $TP$ can be decomposed into horizontal and vertical components, $E=E^\HH + E^\VV$, defined through \eqref{DefinitionHorizontalDistribution}. The projection of the curve $\gamma(s)$ to the base $M$ is denoted by $\gamma_M(s)$, so that $\gamma_M = \pi \circ \gamma$. The Levi-Civita connection on $M$ determined by the projected metric $g_M$ is denoted by $\nabla^M$. 

With this notation, the equations in \cite{ONeill} for the horizontal and vertical components of the covariant derivative $\nabla_{\dot{\gamma}} E$ along $\gamma (s)$ can be written as:
\bal \label{DecompositionParallelTransport}
g_P \big(\nabla_{\dot{\gamma}} E , \, V \big) \ = \ &  g_P \big( \nabla_{\dot{\gamma}} \, E^\VV , \, V \big)   \, +\, \frac{1}{2}\, F_A^a(\dot{\gamma}_M , \, E_M)\ g_P (  e_a ,  \,  V )   \, - \,   g_P \big( S_V \, \dot{\gamma}^\VV, \, E^\HH \big)      \nonumber   \linebr
g_P \big(\nabla_{\dot{\gamma}} E , \, Z \big) \ = \ &  g_M \big( \nabla^M_{\dot{\gamma}_M} \, E_M ,\, \pi_\ast Z \big)   \, +\, \frac{1}{2}\, F_A^a(\pi_\ast Z,\, \dot{\gamma}_M )\ g_P \big(  e_a ,  \,  E^\VV \big)  \, \nonumber \linebr 
&+\, \frac{1}{2}\, F_A^a(\pi_\ast Z,\, E_M )\ g_P (  e_a ,  \,  \dot{\gamma} ) \, +\, g_P \big( S_{\dot{\gamma}^\VV} \, E^\VV, \, Z \big)    \ 
\end{align}
Here $Z$ and $V$ are any horizontal and vertical vectors in $TP$, respectively. We have also translated O'Neill's notation to the one used in this paper, as described in the remark of appendix \ref{Auxiliary results}. As in section \ref{GeneralGeodesicsSubmersions}, the tensor $S$ is the second fundamental form of the fibres of $P$. It is essentially coincides with our $\dd^A  g_K$, as expressed in \eqref{IdentitySecondFundForm}.

The general equations \ref{DecompositionParallelTransport} reduce to \eqref{DecompositionGeodesicEquation} when we take $E= \dot{\gamma}$. The equations satisfied by a parallelly transported vector field $E_{\gamma (s)}$ can be obtained from \ref{DecompositionParallelTransport} by setting $\nabla_{\dot{\gamma}} E = 0$ on the left-hand sides.

\end{appendices}

\newpage

%%%%%%%%%%%%%%%%%%%%%%%%%%%%%%%%%%%%%%%%%%%%%%%%%%%%%%%%%%%%%%

\renewcommand{\baselinestretch}{1.2}

\vspace{1cm}

\end{document}